%% file: main.tex
\documentclass{IEEEtaes_edit}

\usepackage{color}
\usepackage{cite}
\usepackage{amsmath,amssymb,amsfonts}
\usepackage{amsthm}
\usepackage{algorithmic}
\usepackage{algorithm}
\usepackage{array}
\usepackage{graphicx}
\usepackage{textcomp}
\usepackage{stfloats}
\usepackage{url}
\usepackage{booktabs}
\usepackage{multirow}
\usepackage{subcaption}
\usepackage{tikz}
\usepackage{bm}
\usepackage{balance}
\usepackage{comment}
\usepackage{hyperref}

\providecommand{\subfloat}[2][]{\subcaptionbox{#1}{#2}}

\jvol{XX}
\jnum{XX}
\jmonth{XXXXX}
\paper{1234567}
\pubyear{2026}
\doiinfo{TAES.2026.Doi Number}
\newtheorem{theorem}{Theorem}
\newtheorem{lemma}[theorem]{Lemma}
\newtheorem{remark}{Remark}
\newtheorem{proposition}{Proposition}

\newcommand{\argmin}{\mathop{\rm arg~min}\limits}
\newlength{\FigWidth}
\providecommand{\BibTeX}{{\rm B\kern-.05em{\sc i\kern-.025em b}\kern-.08em
    T\kern-.1667em\lower.7ex\hbox{E}\kern-.125emX}}

\begin{document}

\title{
Decentralized Power-Optimal Coordination for Spacecraft Swarms Using Time-Varying Magnetorquer Actuation
}

\author{Yuta Takahashi}
\affil{Department of Mechanical Engineering, Institute of Science Tokyo, Tokyo, Japan}

\author{Shin-ichiro Sakai}
\affil{Department of Spacecraft Engineering, Japan Aerospace Exploration Agency, Kanagawa, Japan}

\receiveddate{Manuscript received ------, 2026; revised ------, ------; accepted ------, ------.\\
This work was not supported by any organization.
}
\authoraddress{Yuta Takahashi is with the Department of Mechanical Engineering, Institute of Science Tokyo, Tokyo, Japan (e-mail: \href{mailto:stateofyuta@gmail.com}{stateofyuta@gmail.com}). 
Shin-ichiro Sakai is with the Department of Spacecraft Engineering, Japan Aerospace Exploration Agency, Kanagawa, Japan (e-mail: \href{mailto:sakai@isas.jaxa.jp}{sakai@isas.jaxa.jp}).
}


\markboth{TAKAHASHI AND SAKAI}{DECENTRALIZED POWER-OPTIMAL COORDINATION FOR SPACECRAFT SWARMS}
\maketitle

\input{main_EMgrouping.tex}

\bibliographystyle{IEEEtran}
\bibliography{references}

\begin{IEEEbiography}[{\includegraphics[width=1in,height=1.25in,clip,keepaspectratio]{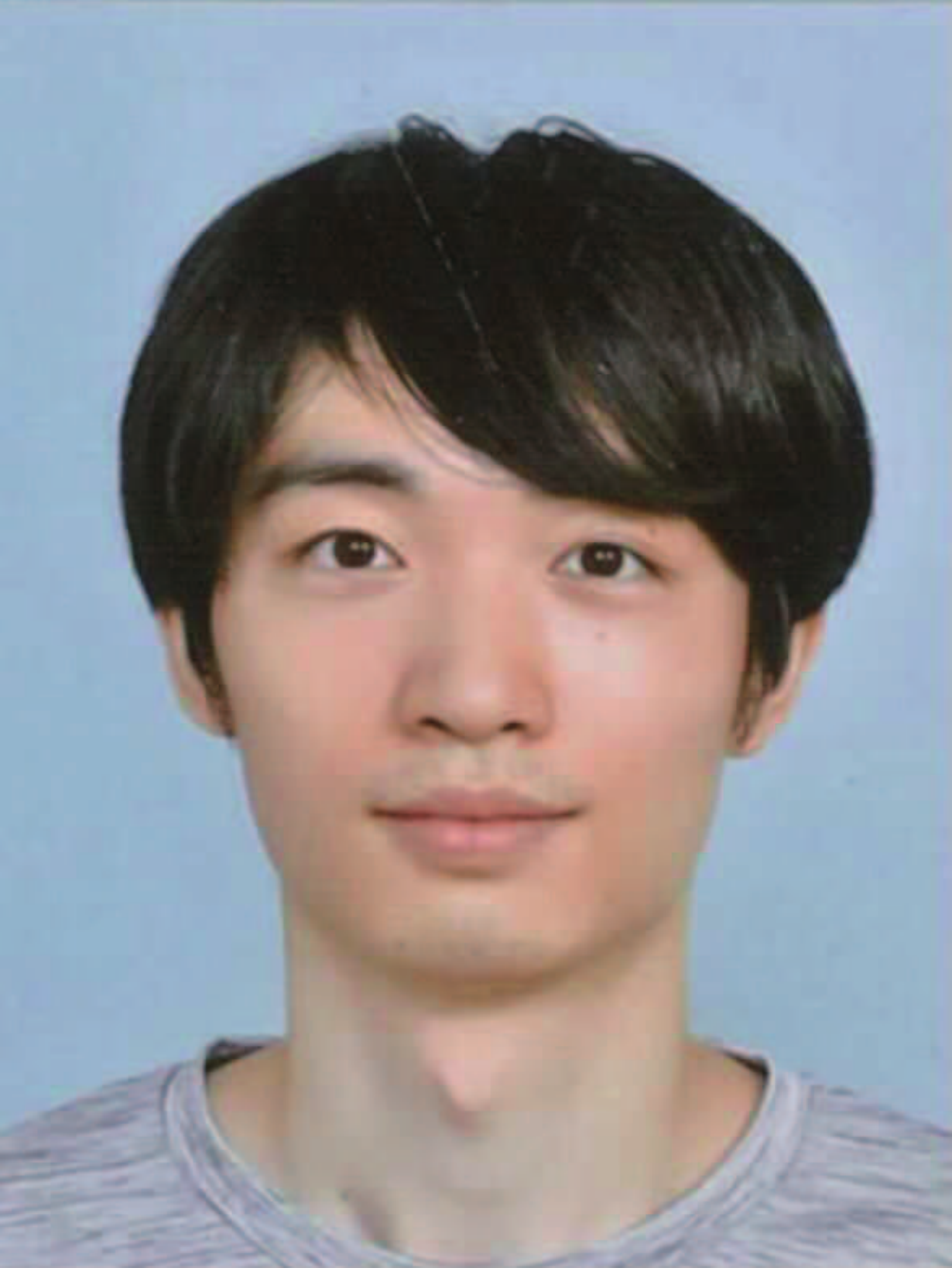}}]{Yuta Takahashi} received his B.S. and M.S. degrees in mechanical engineering (ME) in 2019 and 2021 (early graduations), respectively, from the Tokyo Institute of Technology, Japan. He was the recipient of the Best Student Awards 2017 and 2018 (school class rank 1 of 7/149) and 2019--2021 (school class rank 1/140) in ME at Tokyo Tech. He was a Visiting Student Researcher at the California Institute of Technology in 2022 and 2023. 

He is currently pursuing a Ph.D. degree in ME at the Institute of Science Tokyo, Japan, under the JAXA/ISAS research program. He has also been working at Astroscale as a spacecraft GNC engineer since 2026, designing GNC algorithms for commercial rendezvous-and-docking missions. His research fields include multi-agent and learning-based nonlinear control theory, with primary applications in distributed space and under-actuated systems. He led the avionics team for the space demonstration mission of the membrane-array antenna device ``5G/HELIOS'' from 2019 to 2022, with launches in 2023 and 2026. He collaborated with 
Interstellar Technologies Inc., providing the spacecraft swarm guidance and control algorithm from 2023 to 2026.  

Mr. Takahashi was the recipient of the Best Paper Finalist Award of the GNC Graduate Student Paper Competition at the AIAA SciTech Forum 2025, Best Paper Nominee of the 2nd Space Robotics Workshop at IEEE SMC-IT/SCC 2025, First Place Award at the MTT-Sat Challenge by the IEEE MTT-S in 2023, Miura prize from The Japan Society of Mechanical Engineers in 2021 (1 of 221/31000), Best Group Presentation Award of the AOTULE Conference, which he joined for his B.S. work in 2019, and two conference awards for collaborative research on the hardware implementation of his machine-learning algorithms.
\end{IEEEbiography}
\begin{IEEEbiography}[{\includegraphics[width=1in,height=1.25in,clip,keepaspectratio]{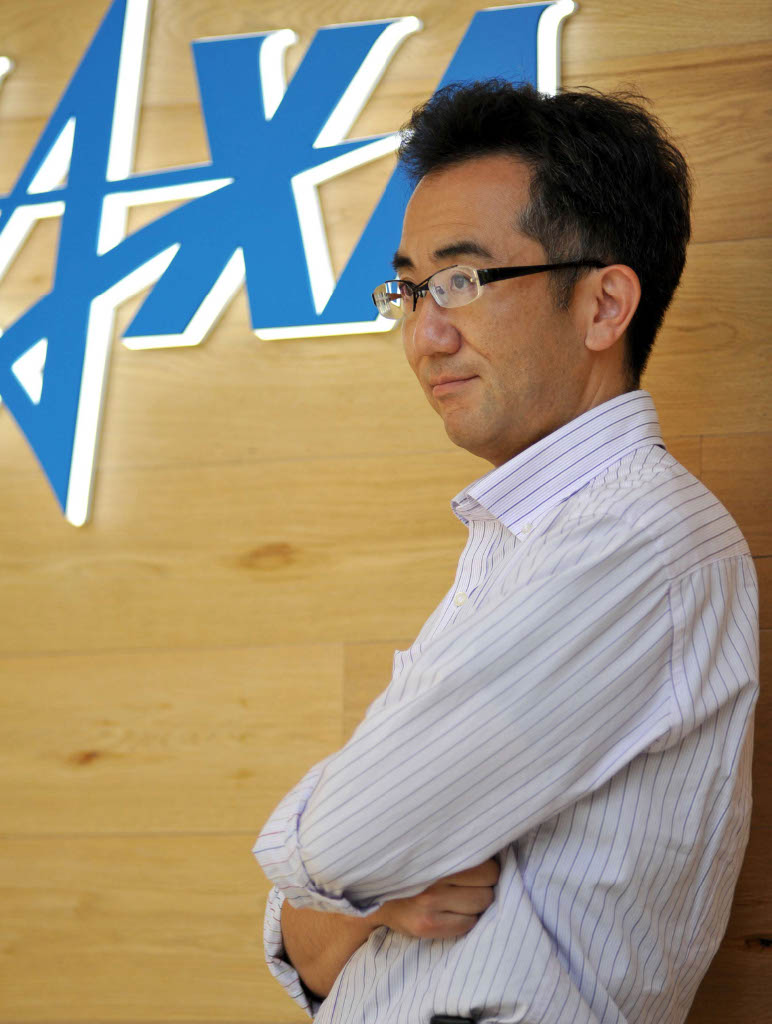}}]{Shin-ichiro Sakai}is a professor at the Institute of Space and Astronautical Science / Japan Aerospace Exploration Agency (JAXA). He received his Ph.D. in Electrical Engineering from the University of Tokyo in 2000. 
He designed attitude control systems such as REIMEI, ASTRO-G, HISAKI, and ARASE. 
He managed JAXA's ``Smart Lander for Investigating Moon'' for a pinpoint moon landing. 
\end{IEEEbiography}
\end{document}

%% file: main_EMgrouping.tex
\begin{abstract}
This paper presents a decentralized power-optimal coordination
framework for magnetically actuated spacecraft swarms. Swarms that form large space structures overcome the
aperture limit set by the launch vehicle and hold their shape on
solar-generated power alone. Magnetic actuation is propellant-free and generated by a magnetorquer, which is commonly used for attitude control. However, every spacecraft interacts with every other within range, and its effect depends on the actuation power and a carrier frequency. We therefore design a decentralized power-optimal framework to jointly derive the interaction graph, frequency grouping, and controller gains. Our decentralized controller preserves angular momentum, which is a nonholonomic constraint. Then, this framework for connected groups whose memberships overlap across carriers guarantees that the relative position
errors, the absolute attitude errors, and the imbalance of the
reaction-wheel momenta converge to the desired states under the decentralized power-optimal allocation. A closed-loop simulation of a thousand spacecraft with the complete alternating-current interaction confirms the framework.  A fast approximate integration with a proven error bound extends the framework to a long-horizon orbital reconfiguration held with high precision.
\end{abstract}
\begin{IEEEkeywords}
Spacecraft swarms, spacecraft formation flying, 
cooperative control, graph partitioning. 
\end{IEEEkeywords}
\begin{figure}[bt!]
\centering
\includegraphics[width=\linewidth]
{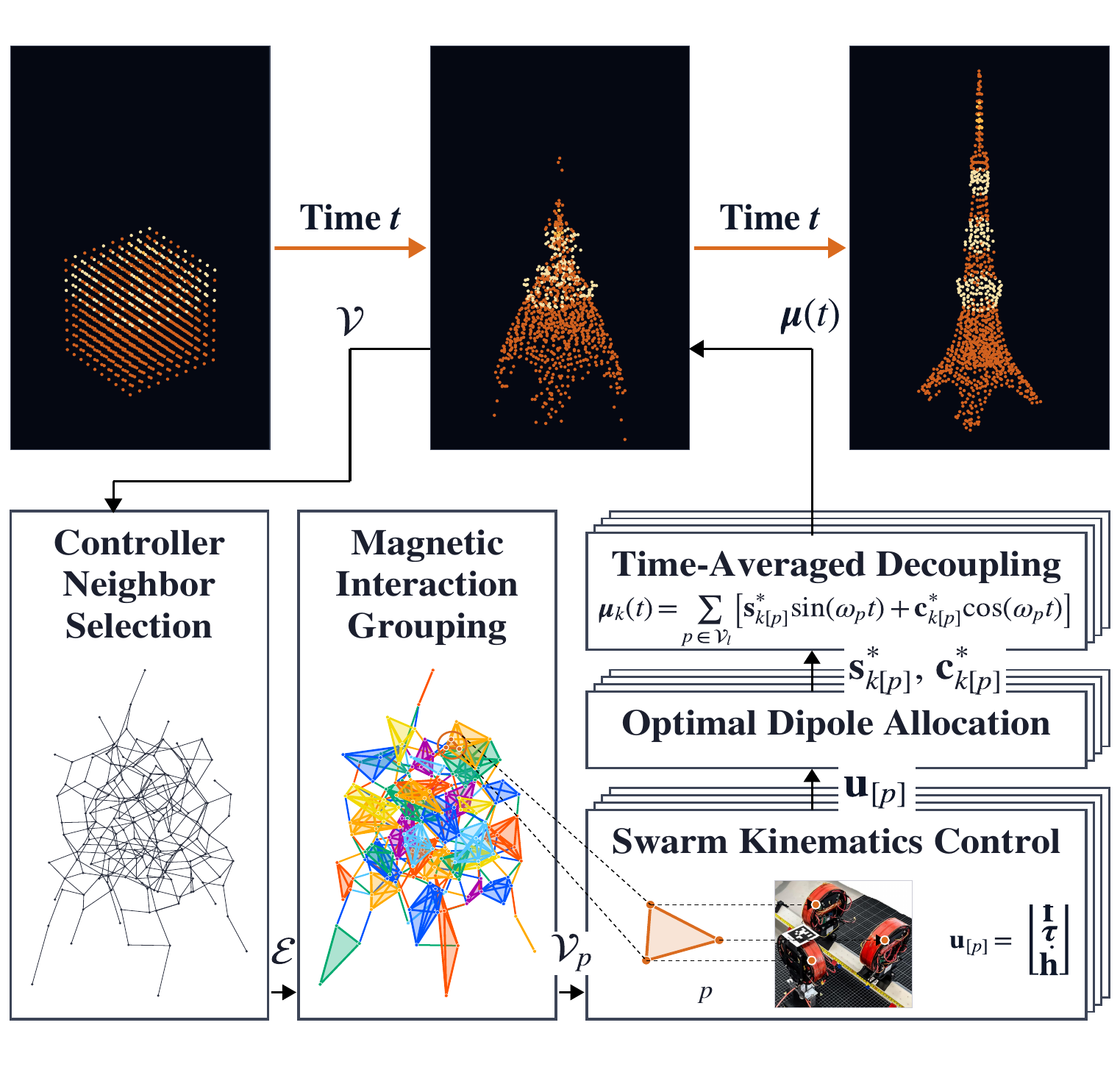}
\caption{Magnetically actuated swarm reconfiguration and system
diagram of grouping and controller design.}
\label{Nonlinear_NODA_MMH_framework_and_grouping}
\end{figure}
\section{Introduction}
\IEEEPARstart{D}{ecentralized} coordination of magnetically actuated
multi-agent systems achieves scalable, propellant-free formation
control 
for the resource constraints of space systems and for high-precision
tracking of the target trajectory. Forming
large space structures from spacecraft swarms is among the most
challenging applications of autonomous multi-agent systems:
distributed apertures overcome the aperture limit set by the launch
vehicle \cite{shim2026optimal,takahashi2026power} and hold their
shape on solar-generated power alone. The actuation has matured as
electromagnetic formation flight (EMFF), which uses the fields of
spacecraft-mounted magnetorquers (MTQs) for relative position and
attitude control \cite{takahashi2022kinematics}. Because the actuator is underactuated, alternating-current (AC) modulation and time-division
control realize six-degree-of-freedom (6-DoF) force and torque control
\cite{youngquist2013alternating,abbasi2022decentralized,takahashi2022kinematics}.

However, the magnetic interaction couples every agent with every
other within range; the unintended interactions degrade the tracking
of the reference trajectory, and the force decays rapidly with
separation and is anisotropic. Since the on-board computing resources
of a space system are limited, an optimization that searches the
grouping of the entire swarm at every control update is undesirable.
Previous studies established a control law that respects the linear-
and nonholonomic angular-momentum conservation
\cite{takahashi2022kinematics} and graph-based coordination for small
formations \cite{takahashi2025distance}.
A simple grouping that secures power optimality and bounds the
interactions between groups within the on-board computation is an open problem.

To this end, this paper proposes a decentralized power-optimal control framework for magnetorquer-driven spacecraft swarms. A coordinator selects the interaction graph and frequency grouping once per reference frame, whereas dipole allocation and each group's control act only on its members' states. Our contributions are threefold. First, the averaged commands of connected overlapping groups preserve the swarm momentum and converge, with the reaction-wheel momenta reaching a uniform distribution, which we verify in a closed-loop simulation with the complete magnetic interaction. Second, a tuning method selects the interaction graph and the
frequency grouping through which the decentralized control acts. The
same partition bounds the unintended interactions between spacecraft
and keeps the loop delay within the requirement, and the
power-minimizing allocation therefore also suppresses the
time-averaging error. Third, an approximate integration that freezes the interaction
geometry over each update interval carries the framework to a
long-horizon orbital reconfiguration with a proven error bound. Together, these results give a decentralized coordination framework whose power optimality and convergence are certified at swarm scale.

The overall structure is as follows. Section~\ref{sec:preliminaries} introduces the magnetic swarm control model, and Section~\ref{sec:eval} formulates the design problem under disturbance using time-averaged error and delay constraints. Section~\ref{sec:three_stage} constructs the decentralized coordination framework including graph selection, edge grouping, and the swarm kinematics controller. Section~\ref{sec:evaluation} demonstrates the proposed framework in spacecraft swarm reconfigurations along with numerically friendly integration methods with a proven error bound. Section~\ref{sec:conclusions} concludes the paper. 
\section{Preliminaries: Magnetic Swarm Control Model}
\label{sec:preliminaries}
\label{NODA_EMFF}%
We first consider one group of $n$ spacecraft, each carrying a
three-axis MTQ, of which $m$ spacecraft carry three-axis reaction
wheels (RWs), $1\leq m\leq n$. Spacecraft $j$ applies the dipole
$\bm\mu_j(t)\triangleq\bm s_j\sin(\omega t)+\bm c_j\cos(\omega t)\in\mathbb R^3$.
First-order averaging over one period gives the electromagnetic force and torque $\bm u_{j\leftarrow k}^{a}\approx\overline{\bm u}_{j\leftarrow k}$
\cite{takahashi2024neural,takahashi2026certified} where
\begin{equation}
\overline{\bm u}_{j\leftarrow k}\triangleq\begin{bmatrix}
    \bm f_{j\leftarrow k}^{a\top}\\
    \bm\tau_{j\leftarrow k}^{a\top}
    \end{bmatrix}
=
\frac{\mu_0}{8\pi}
Q_{j\leftarrow k}
\left(
\bm s_k\otimes\bm s_j+
\bm c_k\otimes\bm c_j
\right),
\label{NODA_unified_averaged_interaction}
\end{equation}
and $Q_{j\leftarrow k}$ is 
the matrix function of $d_{j\leftarrow k}=\|\bm r_{j\leftarrow k}\|$ and the line-of-sight
(LOS) matrix $C^{A/L_{j\leftarrow k}}$ (see \cite{takahashi2024neural}):
\begin{equation}
\begin{aligned}
Q_{{j\leftarrow k}}
&=
\begin{bmatrix}
C^{A/L_{j\leftarrow k}}\Psi_f
\\
C^{A/L_{j\leftarrow k}}\Psi_\tau
\end{bmatrix}
(C^{L_{j\leftarrow k}/A}\otimes C^{L_{j\leftarrow k}/A}),
\end{aligned}
\label{EMgrouping_dipolemodel}
\end{equation}
\begin{equation*}
\begin{aligned}
\Psi_f
&=
\frac{1}{d_{j\leftarrow k}^4}
{\small\begin{bmatrix}
-6&0&0&0&3&0&0&0&3\\
0&3&0&3&0&0&0&0&0\\
0&0&3&0&0&0&3&0&0
\end{bmatrix}},\\
\Psi_\tau
&=
\frac{1}{d_{j\leftarrow k}^3}
{\small\begin{bmatrix}
0&0&0&0&0&1&0&-1&0\\
0&0&2&0&0&0&1&0&0\\
0&-2&0&-1&0&0&0&0&0
\end{bmatrix}}.
\end{aligned}
\end{equation*}
Since magnetic interactions are internal, $\overline{\bm u}_{j\leftarrow k}$ preserves the linear momentum $\bm P$ and the angular momentum $\bm L$. 
\begin{lemma}[Kinematics Control \cite{takahashi2022kinematics}]
\label{NODA_theorem_experimental_controller}
We define the linear- and angular-momentum conservation as  $A(n,m;\bm q)\bm\zeta=[{}^I\bm P^\top,{}^I\bm L^\top]^\top$
where $A(n,m;\bm q)\in\mathbb R^{6\times(6n+3m)}$ and $\bm\zeta=[\dot{\bm r}^{i\top},\bm\omega^{b\top},
\bm h^{b\top}]^{\top}$. 
Then, arbitrary command $\bm u_c=B(n,m)^{-1}M(n,m)S(n,m)\mathfrak u_c$,
$\mathfrak u_c\in\mathbb R^{6n+3m-6}$, preserves angular momenta and is commandable where $S(n,m)$ is the null space of $A(n,m)$ and $M(n,m)$ and
$B(n,m)$ are the mass--inertia and input matrix.
\end{lemma}
\label{sec:allocation}
For multiple groups comprising $N\gg n$ spacecraft, let the leader
set $\mathcal V_l$ index the groups, $\mathfrak g_p$ consisting of the
spacecraft that share the carrier $\omega_p$; the $k$th dipole is \cite{takahashi2024neural}
\begin{equation}
\bm\mu_k(t)
=
\sum_{p\in\mathcal V_l}
\left[
\bm s_{k[p]}\sin(\omega_pt)
+
\bm c_{k[p]}\cos(\omega_pt)
\right].
\label{NODA_multi_frequency_dipole_moment}
\end{equation}
Groups on different frequencies decouple in the first-order averaged dynamics, and this realizes swarm control \cite{takahashi2022kinematics,takahashi2024neural}.
The minimum-power
allocation problem is
\begin{equation}
\begin{aligned}
W^{\star}(\mathfrak g)
&\triangleq
\min_{\bm s_{\mathfrak g},\bm c_{\mathfrak g}}
\quad
\left(
\bm s_{\mathfrak g}^{\top}\bm s_{\mathfrak g}
+
\bm c_{\mathfrak g}^{\top}\bm c_{\mathfrak g}
\right)/2\\
\mathrm{s.t.}\quad
\bm u_{cj}
&=
\frac{\mu_0}{8\pi}
\sum_{k\in\mathcal \mathfrak{g}_p}
Q_{j\leftarrow k}
\left(
\bm s_k\otimes\bm s_j+
\bm c_k\otimes\bm c_j
\right),
\end{aligned}
\label{Nonlinear_NODA_MMH_power_allocation}
\end{equation}
with 
$\bm u_{cj}$ the command input. Its Lagrange dual $J_{\mathrm{dual}}(\mathfrak g)$ is a semidefinite program in the multipliers of the $n$ wrench constraints \cite{takahashi2024neural}. The power
$W_{\mathrm{ODA}}(\mathfrak g)$ of the implemented optimal dipole
allocation (ODA) then satisfies
$J_{\mathrm{dual}}\leq W^{\star}\leq W_{\mathrm{ODA}}\leq\gamma_WJ_{\mathrm{dual}}$, where $\gamma_W\geq1$ is the duality-gap ratio, the certificate ratio with $W_{\mathrm{ODA}}\leq\gamma_WJ_{\mathrm{dual}}$ on every allocation.
The inverse baseline \cite{takahashi2024neural,takahashi2026certified}
fixes one endpoint amplitude, inverts $Q_{j\leftarrow k}$ for the other, and keeps the best of $K$ random initializations (best-of-$K$ inverse).
\section{Problem Formulation}
\label{sec:eval}
For a
given feasible reference trajectory $P_r(t)$, this work accommodates general guidance and trajectory-generation
laws, and we adopt a network-based PD controller as a representative example:
\begin{equation}
\begin{aligned}
F(t)&=-k_pL(P-P_r)-k_vL(\dot P-\dot P_r),\\
T(t)&=-k_\sigma L\Sigma-k_\omega L\Omega,
\end{aligned}
\label{centralized_controller}
\end{equation}
where $P$, $\Sigma$, and $\Omega$ stack the positions, attitude
parameters, and angular velocities of the swarm. We define an undirected graph as $\mathcal G=(\mathcal V,\mathcal E)$, $N=|\mathcal V|$, that has
neighborhoods $\mathcal N_j=\{k\mid(j,k)\in\mathcal E\}$, adjacency $\widetilde A$, and Laplacian $L=\operatorname{Diag}(\widetilde A\bm1)-\widetilde A$ with eigenvalues $0=\lambda_1(L)\leq\cdots\leq\lambda_N(L)$ \cite{mesbahi2010graph}. 
\subsection{Cross-Group Disturbances and Loop Delay}
\label{sec:disturbance_delay}
The swarm reconfiguration command, such as \eqref{centralized_controller}, is realized by the multi-frequency dipoles (\ref{NODA_multi_frequency_dipole_moment}) and wheel torques, and leaves the time-averaging error of \eqref{NODA_unified_averaged_interaction} as a disturbance. For the dipole
\eqref{NODA_multi_frequency_dipole_moment} and channels $p,q$, the
complete interaction is
\begin{equation}
\bm u_{j\leftarrow k}(t)
=
\overline{\bm u}_{j\leftarrow k}
+
\sum_{p\in\mathcal V_l}
\bm d_{j\leftarrow k}^{2\omega_p}(t)+
\sum_{\substack{p,q\in\mathcal V_l\\p\neq q}}
\bm d_{j\leftarrow k}^{pq}(t),
\label{Nonlinear_NODA_MMH_multifrequency_complete_interaction}
\end{equation}
with the intra-group disturbance $\bm d_{j\leftarrow k}^{2\omega_p}$
at twice the channel frequency \cite{takahashi2024neural,takahashi2026certified} and the cross-group disturbance
$\bm d_{j\leftarrow k}^{pq}$ at
$\Delta\omega_{pq}=\omega_p-\omega_q$ and
$\Sigma\omega_{pq}=\omega_p+\omega_q$,
\begin{equation}
\begin{gathered}
\bm d_{j\leftarrow k}^{pq}(t)
=\frac{\mu_0}{8\pi}Q_{j\leftarrow k}
\big[
\bm A_{jk}^{pq}\cos(\Delta\omega_{pq}t)
+\bm B_{jk}^{pq}\sin(\Delta\omega_{pq}t)\\
+\bm C_{jk}^{pq}\cos(\Sigma\omega_{pq}t)
+\bm D_{jk}^{pq}\sin(\Sigma\omega_{pq}t)
\big],\\
{\scriptstyle
\bm A_{jk}^{pq}=\bm c_{k[p]}\otimes\bm c_{j[q]}
+\bm s_{k[p]}\otimes\bm s_{j[q]},\ \
\bm B_{jk}^{pq}=\bm s_{k[p]}\otimes\bm c_{j[q]}
-\bm c_{k[p]}\otimes\bm s_{j[q]},}\\
{\scriptstyle
\bm C_{jk}^{pq}=\bm c_{k[p]}\otimes\bm c_{j[q]}
-\bm s_{k[p]}\otimes\bm s_{j[q]},\ \
\bm D_{jk}^{pq}=\bm s_{k[p]}\otimes\bm c_{j[q]}
+\bm c_{k[p]}\otimes\bm s_{j[q]}.}
\end{gathered}
\label{eq:full_AC_pairwise_interaction}
\end{equation}
Both errors have zero mean over a period \cite{takahashi2022kinematics} and $\sup_t\|\bm d_{j\leftarrow k}^{2\omega_p}(t)\|$ is bounded by a geometry and a power
\cite{takahashi2024neural},
and we prove the counterpart for
$\bm d_{j\leftarrow k}^{pq}(t)$.
\begin{lemma}
\label{NODA_MIG_cross_frequency_bound}
Consider a grouping $Z=\mathfrak g_{[1,N_g]}$ in
which group $\mathfrak g_\ell$ carries $\omega_\ell$ and
has the dual cost $J_\ell=J_{\mathrm{dual}}(\mathfrak g_\ell)$. Then, the cross-group disturbance that $\mathfrak g_q$ exerts on $\mathfrak g_\ell$, $\mathcal D_{\ell\leftarrow q}^{\mathrm{AC}}
=\sum_{j\in\mathfrak g_\ell}\sup_t
\Big\|\sum_{k\in\mathfrak g_q}\bm d_{j\leftarrow k}^{\ell q}(t)\Big\|_2$, summed as
$\mathcal D_{\ell\leftrightarrow q}^{\mathrm{AC}}
=\mathcal D_{\ell\leftarrow q}^{\mathrm{AC}}
+\mathcal D_{q\leftarrow\ell}^{\mathrm{AC}}$, obeys
\begin{equation}
\sum_{\ell<q}
\mathcal D_{\ell\leftrightarrow q}^{\mathrm{AC}}
\leq
\frac{\mu_0}{2\pi}
\gamma_W\Phi_{\mathrm{AC}}(Z)
\triangleq
\frac{\mu_0}{2\pi}
\gamma_W
\sum_{\ell<q}
\Xi_{\ell q}(Z)(J_\ell+J_q),
\label{eq:MIG_AC_cut_penalty}
\end{equation}
where $\Xi_{\ell q}(Z)\triangleq
[\sum_{j\in\mathfrak g_\ell}\sum_{k\in\mathfrak g_q}
\gamma_Q^2(d_{jk})]^{1/2}$ is the root-sum-square of the distance
envelope $\gamma_Q(d)\triangleq\sup\{\|Q(\bm r)\|_2\mid\|\bm r\|_2\geq d\}$
over the cross-group pairs.
\end{lemma}
\begin{proof}
With $p_{j[\ell]}=\|\bm s_{j[\ell]}\|^2+\|\bm c_{j[\ell]}\|^2$, the
amplitudes $\bm A_{jk}^{\ell q},\ldots,\bm D_{jk}^{\ell q}$ of
\eqref{eq:full_AC_pairwise_interaction} have norm at most
$\sqrt{p_{j[\ell]}p_{k[q]}}$, and each pairwise term is bounded by $\frac{\mu_0}{4\pi}\gamma_Q(d_{jk})\sqrt{p_{j[\ell]}p_{k[q]}}$. Summing over the two groups by the triangle inequality
and applying the Cauchy--Schwarz inequality gives
\begin{equation*}
\mathcal D_{\ell\leftarrow q}^{\mathrm{AC}}
\leq\frac{\mu_0}{4\pi}\Xi_{\ell q}
\Bigl(\sum_{j\in\mathfrak g_\ell}p_{j[\ell]}
\sum_{k\in\mathfrak g_q}p_{k[q]}\Bigr)^{1/2}
=\frac{\mu_0}{2\pi}\Xi_{\ell q}\sqrt{W_\ell W_q},
\end{equation*}
where $\sum_{j\in\mathfrak g_\ell}p_{j[\ell]}=2W_\ell$ with
$W_\ell=W_{\mathrm{ODA}}(\mathfrak g_\ell)$ from the objective of
\eqref{Nonlinear_NODA_MMH_power_allocation}. The arithmetic-geometric
mean inequality and $W_\ell\leq\gamma_WJ_\ell$ replace
$2\sqrt{W_\ell W_q}$ by $\gamma_W(J_\ell+J_q)$, and summing over all group pairs yields
\eqref{eq:MIG_AC_cut_penalty}.
\end{proof}
The multi-frequency actuation holds the amplitudes over an update interval $T_u=N_TT$, which acts on the consensus loop as a zero-order hold, that is, a pure delay of $T_u/2$. The classical delay condition of consensus
\cite{olfati2004consensus,mesbahi2010graph}
bounds the velocity feedback rate $k_\lambda=k_v/m_{\mathrm s}$ by
$k_\lambda\lambda_N\tau<\pi/2$ under the loop delay $\tau$ such that
\begin{equation}
k_\lambda\lambda_2<(\pi/2\tau)\lambda_2/\lambda_N.
\label{eq:delay_bound}
\end{equation}
The loop delay $\tau$ is the sum of the communication latency
$\tau_{\mathrm{comm}}$ and the hold delay, which is $\pi/\omega_0$ at
$N_T=1$. Because $N_f$ commensurate carriers must fit below the
highest admissible carrier $\omega_{\max}$, the base frequency obeys
$\omega_0\leq\omega_{\max}/N_f$, and the loop stays below its measured
divergence delay $\tau^{\mathrm{div}}$ only if
$N_f<(\tau^{\mathrm{div}}-\tau_{\mathrm{comm}})\,\omega_{\max}/\pi$.
\subsection{Decentralized Power-Optimal Grouping}
\label{sec:problem_formulation}
\label{sec:coordination}
The design freedom is therefore the commanded interaction graph $\widetilde A$ on the given node set $\mathcal V$ and its partition $Z$ into frequency groups, 
for the power-optimal grouping
\begin{equation}
\begin{aligned}
&(\widetilde A^\star,Z^\star) =\argmin_{\widetilde A,\ Z\in\mathcal Z_{n_{\max}}(\widetilde A)}\
\sum_{\mathfrak g\in Z}W^{\star}(\mathfrak g)\\
&\qquad\mathrm{s.t.}\quad
N_f(Z)<\bigl(\tau^{\mathrm{div}}(\widetilde A)-\tau_{\mathrm{comm}}\bigr)
\,\omega_{\max}/\pi,
\end{aligned}
\label{eq:joint_grouping_problem}
\end{equation}
where $N_f(Z)$ is the number of distinct carriers per neighborhood that the partition requires, $\mathcal Z_{n_{\max}}(\widetilde A)$ collects the partitions of $\widetilde A$ into connected groups of at most $n_{\max}$ spacecraft. 
\section{Decentralized Coordination Framework}
\label{sec:three_stage}
\label{sec:criteria}
\label{sec:grouping_power_trend}
\label{sec:example}
We solve the design problem \eqref{eq:joint_grouping_problem} in three stages: 
\begin{enumerate}
  \item Controller neighbor selection (CNS) selects the delay-admissible nearest-neighbor graph
  \item Magnetic interaction grouping (MIG) partitions
its nodes into frequency groups along CNS's edges 
  \item Swarm kinematics controller (SKC) derives the force-and-torque commands for MIG's groups 
\end{enumerate}
\label{sec:numerical_results_1000sat}
We quantify all three stages on a $1000$-spacecraft reconfiguration that transits from a dense trunk to the sparse Tokyo-Tower formation in $T_{\mathrm{end}}=60$~s (Fig.~\ref{Nonlinear_NODA_MMH_framework_and_grouping}).
\label{sec:constructions}
\subsection{Stage 1: Controller Neighbor Selection for Nodes}
\label{sec:cns_stage}
CNS fixes the commanded graph $\widetilde A$ through three criteria
to reduce the actuation power, the time-averaging error, and the
consensus time.
\subsubsection{Criterion 1: Fourth-Power Distance Law}
\label{sec:distance_dominant}
The first criterion admits short edges. The dipole model
$Q_{j\leftarrow k}$ of \eqref{EMgrouping_dipolemodel} decays as
$d_{jk}^{-4}$, and a longer edge therefore demands larger dipoles and
a larger actuation power $W^\star$ in
\eqref{Nonlinear_NODA_MMH_power_allocation}. On the reference mission, pairing each spacecraft with its $30$th-nearest neighbor rather than its nearest raises the median paired distance from $0.65$ to $1.39$~m and the mission energy $7.3$ times in Fig.~\ref{fig:model_norm_blocks}(a).
Note that this model is accurate to a few percent in the median at the flight scale of
Fig.~\ref{fig:model_norm_blocks}(b), where $R_{\mathrm{coil}}$ is the
coil radius and $d/2R_{\mathrm{coil}}\approx3$ at the median paired
distance for $R_{\mathrm{coil}}=0.1$~m.
\subsubsection{Criterion 2: Delay-and-Hold Divergence Threshold}
\label{sec:spectral_tradeoff}
The second criterion floors the number of neighbors under the loop delay $\tau$ of \eqref{eq:delay_bound}. At a consensus stiffness $k_p\lambda_2$ common to the candidate graphs, the delayed feedback of the fastest mode scales with $\lambda_N/\lambda_2$, which falls as the graph densifies. We measure the divergence delay $\tau^{\mathrm{div}}$ of the linear part of \eqref{eq:group_error_dynamics}, the consensus loop $\ddot{\bm x}=-k_pL_k\bm x(t-\tau)-k_vL_k\dot{\bm x}(t-\tau)$ on the $k$-nearest-neighbor ($k$-NN) graph with Laplacian $L_k$ at that common stiffness, and take its worst value over the mission geometries. A $k$-NN graph is admissible as $\widetilde A$ if this worst $\tau^{\mathrm{div}}$ exceeds the loop delay of the implementation. Optimizing edge selection does not scale because orienting the Delaunay-triangulation graph toward the leaders via a linear program grows as $k^{4.7}$ in the neighbor count, versus $k^{1.9}$ for the neighborhood-greedy construction of the $k$-NN graph.
\begin{figure}[tb!]
\centering
\subfloat[Mission dipole energy for pairings on $k{=}1$--$30$ nearest neighbors, labeled by the median distance between paired spacecraft.]{
\begin{minipage}{0.97\linewidth}\centering
\includegraphics[width=\linewidth]
{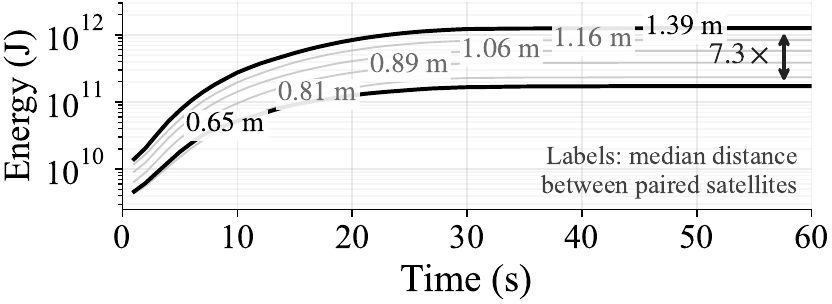}
\end{minipage}
\label{fig:neighbor_k_energy}}\\
\vspace{0.02cm}
\subfloat[Exact-to-dipole force and torque gains.]{
\includegraphics[width=0.97\linewidth]
{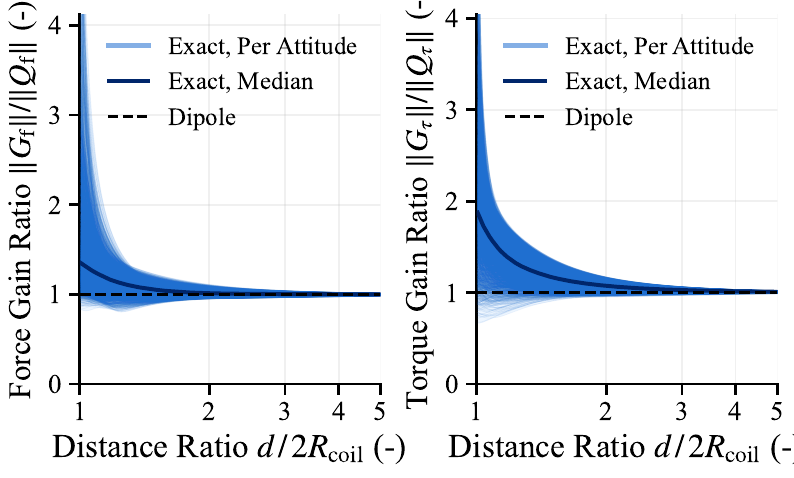}
\label{fig:model_norm_force}\label{fig:model_norm_torque}}\\
\vspace{0.02cm}
\subfloat[Worst-case delay margin.]{
\includegraphics[width=0.97\linewidth]
{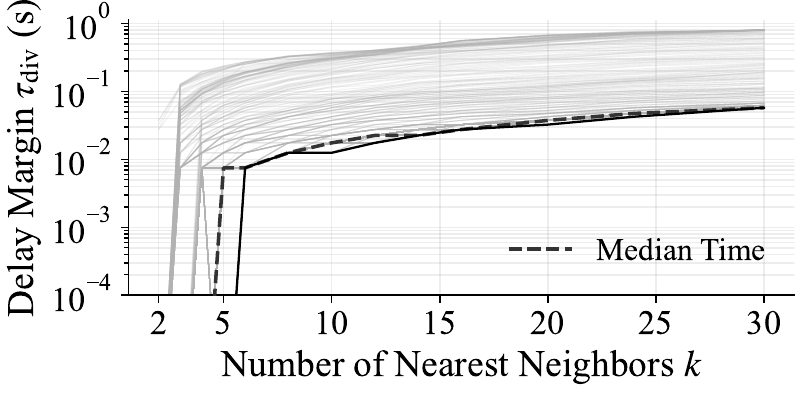}
\label{fig:delay_margin_k}}
\caption{Mission energy over the neighbor distance, finite-coil correction with a minimum ratio of $0.80$ in force and $0.67$ in torque, and delay margin.}
\label{fig:model_norm_blocks}
\end{figure}
\subsubsection{Criterion 3: Information Graph for Consensus}
\label{sec:info_graph}
The third criterion shortens the consensus time with long information
edges. The consensus rate of \eqref{centralized_controller} is set by $\lambda_2$ of the graph that exchanges the
states, and long edges raise $\lambda_2$, whereas realizing them
magnetically is prohibitive because the interaction of
\eqref{EMgrouping_dipolemodel} decays as $d^{-4}$. We therefore form
\eqref{centralized_controller} on the information graph, the commanded graph $\widetilde A$ with long edges added, and these edges
carry no magnetic force. Any set of long edges serves this purpose, and a fixed set is used throughout. The per-spacecraft forces $\bm F$ of \eqref{centralized_controller} sum to
zero over the swarm, and on a connected $\widetilde A$ with incidence matrix $E$ they are realized exactly by the minimum-norm edge flow
$\bm f=E^{+}\bm F$ with $E\bm f=\bm F$, solved once per update interval.
\subsection{Stage 2: Magnetic Interaction Grouping for Edges}
\label{sec:mig_stage}
MIG partitions the swarm into frequency groups of at
most $n_{\max}$ to reduce the power and the number of carriers.
\begin{algorithm}[tb!]
\begin{algorithmic}[1]
\STATE \textbf{Input}
$\bm r_N,\bm u_N,\widetilde A,n_{\max},\bar\rho,R^\star$, and
$\omega_0$ such that each reference frame spans an integer number of common periods $T$; set the force cutoff and
same-color separation to the effective radius $R^\star$. Form the minimum-weight no-singleton pairing using
$d_{\mathrm{eff},jk}$
\FOR{$n=3,\ldots,n_{\max}$}
  \REPEAT
    \STATE Form disjoint neighborhoods of nearby groups with at most $2n$ spacecraft
    \FOR{each neighborhood $\mathcal S$ \textbf{in parallel}}
      \STATE Enumerate connected partitions $\pi$ satisfying
      $2\leq|\mathfrak g|\leq n$ and $\rho(\mathfrak g)\leq\bar\rho$
      \STATE Minimize the far-field dual cost
      $\sum_{\mathfrak g\in\pi}J_{\mathrm{dual}}(\mathfrak g)$
      and commit only a decrease
    \ENDFOR
  \UNTIL{no neighborhood reduces the cost}
\ENDFOR
\STATE \textbf{return} the partition $Z$
\end{algorithmic}
\vspace{0.5mm}
\hrule
\vspace{1.2mm}
\centering
\includegraphics[width=\linewidth]
{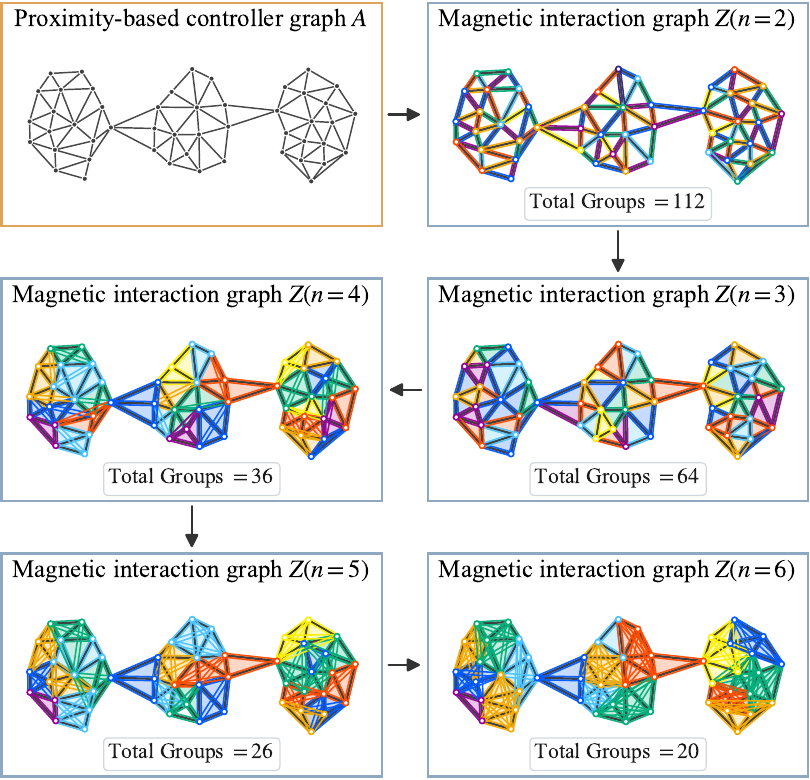}
\caption{Agglomerative grouping.}
\label{alg:agglomerative_grouping}
\label{alg:pipeline}
\end{algorithm}
\subsubsection{Effective-Distance Score and Agglomeration}
\label{sec:mig_edge_distance}
\label{sec:anisotropy}
\label{sec:cancellation}
\label{sec:score}
We score a candidate edge by its effective distance
$d_{\mathrm{eff},jk}$. The power of an edge grows with the fourth
power of the distance and linearly with the demanded wrench in
\eqref{EMgrouping_dipolemodel}. We therefore define the effective
distance as the distance at which the reference wrench
$u_{\mathrm{ref}}$ costs the same power as the demanded wrench at
$d_{jk}$,
\begin{equation}
d_{\mathrm{eff},jk}
=
d_{jk}
\frac{\|\widetilde{\bm u}_{jk}\|^{1/4}}{u_{\mathrm{ref}}^{1/4}}
,\
\widetilde{\bm u}_{jk}=
\begin{bmatrix}
\bm f_{jk}/\gamma_{\mathrm f,jk}\\
\bm\tau_{jk}/(d_{jk}\gamma_{\tau,jk})
\end{bmatrix}
\label{eq:deff_pooling}
\end{equation}
where $\gamma_{\mathrm f,jk}=\|d_{jk}^{4}Q_{j\leftarrow k}[1\!:\!3,:]\|_2$
and $\gamma_{\tau,jk}=\|d_{jk}^{3}Q_{j\leftarrow k}[4\!:\!6,:]\|_2$ are
the gains of the force and torque blocks, and the torque term has the
extra $d_{jk}$ of its $d^{-3}$ law. The alignment changes the power by at most a factor of two. With $\rho$ the cosine of the angle between the
demanded force and the line of sight, the power of a pair is
$P/P(\rho{=}1)=2-\rho^{2}\leq2$. This factor becomes at most
$2^{1/4}\approx1.19$ in $d_{\mathrm{eff}}$ and $2^{1/3}$ on torque.
The exact alignment is evaluated in the dual cost at every merge. We
also use this score for groups, since a group routes its demand along
the lines of sight and does not cancel it. For existing groups, we
define
$d_{\mathrm{eff}}(\mathfrak g_p,\mathfrak g_q)
=\min\{d_{\mathrm{eff},jk}\mid
j\in\mathfrak g_p,\,k\in\mathfrak g_q,\,
(j,k)\in\mathcal E(\widetilde A)\}$ and the group diameter
$\rho(\mathfrak g)=\max_{j,k\in\mathfrak g}d_{jk}$. Algorithm~\ref{alg:agglomerative_grouping} grows the groups with this
score. We first form the minimum-weight perfect matching
$Z_{\mathrm{match}}$ on $d_{\mathrm{eff},jk}$ by the blossom algorithm
\cite{kolmogorov2009blossom} in $O(N^3)$ time. We then merge
neighboring groups only when the dual cost
$J(Z)=\sum_{\mathfrak g\in Z}J_{\mathrm{dual}}(\mathfrak g)$
decreases. The candidate set of a group $\mathfrak g$ is the connected
repartitions of $\mathfrak g$ and its neighboring groups in
$\widetilde A$ into groups of size $2$ to $n$ and diameter at most
$\bar\rho$. A partition is locally stationary when
no candidate decreases $J$.
\begin{proposition}
\label{prop:MIG_descent_termination}
\label{prop:matching_certificate}
Consider Algorithm~\ref{alg:agglomerative_grouping} that commits only
the repartitions that decrease $J$. Then the algorithm terminates after
finitely many repartitions at a locally stationary partition
$Z_{\mathrm{final}}$ with $J(Z_{\mathrm{final}})\leq J(Z_{\mathrm{match}})$.
Moreover, the ODA satisfies
$W_{\mathrm{ODA}}(Z_{\mathrm{final}})\leq\gamma_W\,J(Z_{\mathrm{final}})$
at every time.
\end{proposition}
\begin{proof}
Every candidate set is finite, and every committed repartition strictly
decreases $J\geq0$. Only finitely many repartitions are committed, which
gives the first claim. Summing the
certificate $W_{\mathrm{ODA}}(\mathfrak g)\leq\gamma_WJ_{\mathrm{dual}}(\mathfrak g)$
of Section~\ref{sec:allocation} over the groups of $Z_{\mathrm{final}}$
yields the second claim.
\end{proof}
\subsubsection{Offline Evaluation of the Grouping Stage}
\label{sec:nn_greedy_distance}
\label{sec:agglomerative}
\label{sec:cumulative_punchline}
\label{sec:eval_peak_ac}
\begin{figure}[tb!]
\centering
\subfloat[Force routing and the alignment penalty.]{
\includegraphics[width=0.92\linewidth]
{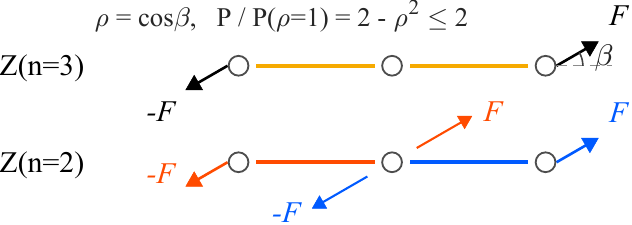}
\label{fig:n2_n3_concept}\label{fig:anisotropy_pair}}\\
\vspace{0.02cm}
\subfloat[Conservatism and demand alignment.]{
\includegraphics[width=0.96\linewidth]
{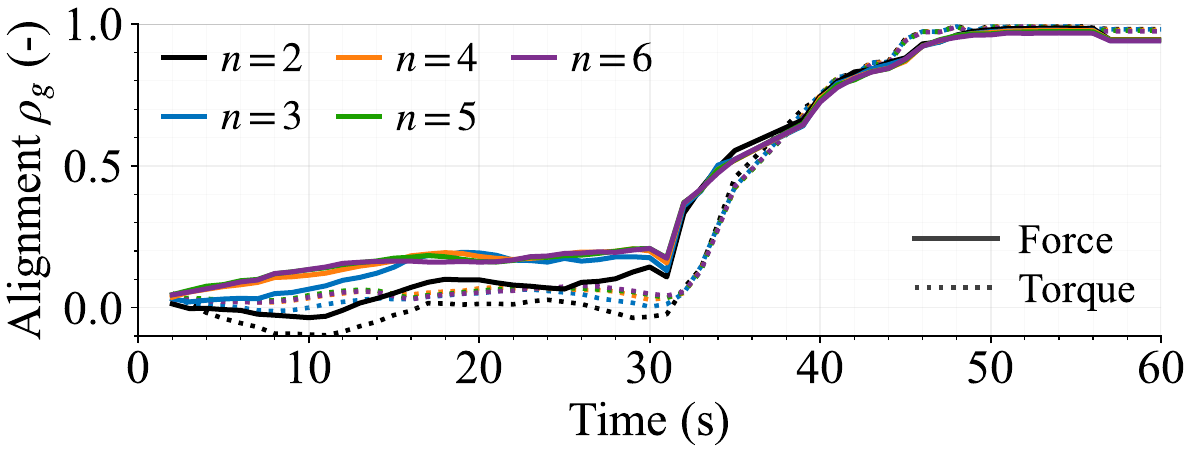}
\label{fig:factor3_cancel_c}\label{fig:factor3_cancel_rho}}
\caption{Grouping mechanisms: (a)~routing and the alignment penalty;
(b)~$\rho_g$ per group size.}
\label{fig:factor3_cancel_combined}
\end{figure}
\begin{figure}[tb!]
\centering
\subfloat[Weight-sweep landscape, where the ratio stays within $1.16$--$1.25\times$ against the bound of $2$.]{
\includegraphics[width=\linewidth]
{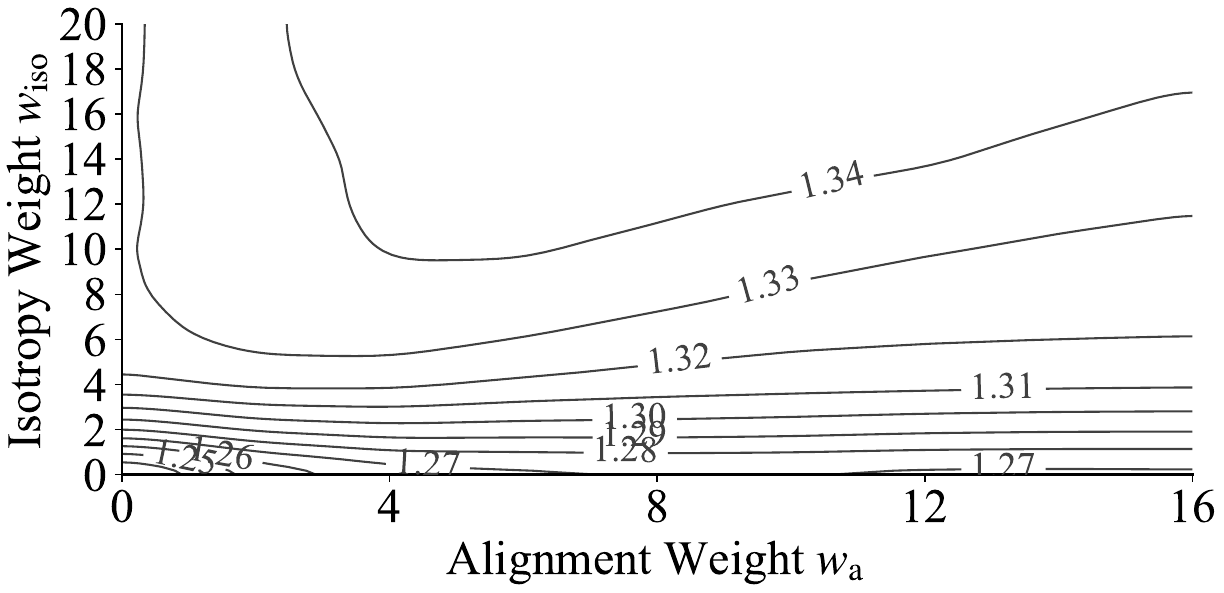}
\label{fig:lambda_landscape_a}}\\
\vspace{0.02cm}
\subfloat[Mission power-ratio distribution.]{
\includegraphics[width=\linewidth]
{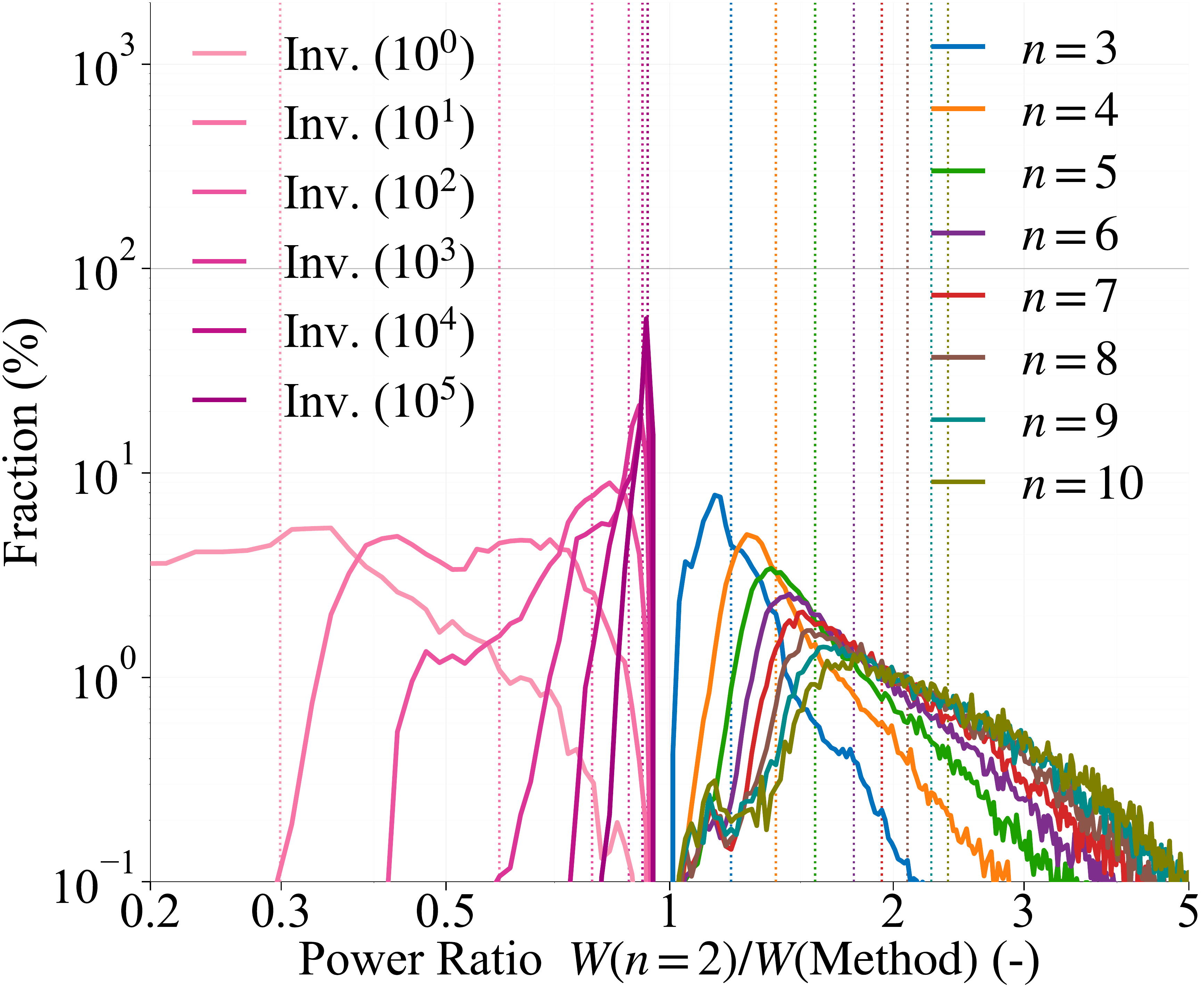}
\label{fig:capsweep_ratio}}
\caption{(a)~Power ratio over the alignment weight $w_{\mathrm a}$. (b)~Mission ratio $W(n{=}2)/W(\mathrm{Method})$.}
\label{fig:lambda_landscape}
\end{figure}
We evaluate the grouping offline on the reference trajectory,
independently of the closed loop. We check three points, namely the
effect of the demand direction on the chosen edges, the routing of
the demands within a group, and the power saved by groups larger than
a pair.

For the first point, we increase the alignment weight $w_{\mathrm a}$
until the result no longer changes, and the grouped power ratio moves
only within $1.16$--$1.25\times$ in Fig.~\ref{fig:lambda_landscape}.
Note that $57\%$ of the competing candidate pairs are within the
$2^{1/4}$ band of each other, and the demand direction therefore
rarely changes which pair is chosen.

For the second point, we measure the demand alignment $\rho_g$, which
rises from $0.07$ in the transient to $0.98$ once the demands align in
Fig.~\ref{fig:factor3_cancel_combined}(b). The direction can then
change the power by at most $2-\rho_g^{2}=1.05$, and the demand is
routed along the lines of sight.

For the third point, we compare groups with pairs on the mission
samples in Table~\ref{fig:density_bench}, where the best-of-$K$
inverse stays below the ODA pairing in the median even at
$K{=}10^{5}$, whereas grouping at $n{=}4$ exceeds the pairing by
$1.39$ in the median and by $15.9$ at the maximum. The gain is
therefore large only where the neighbors are close and aligned, and
we limit the size of a group by its diameter $\bar\rho$ in
Algorithm~\ref{alg:agglomerative_grouping}. Note that larger groups
reduce the number of groups but raise the cross-group bound
$\Phi_{\mathrm{AC}}$ of \eqref{eq:MIG_AC_cut_penalty}, whose
$d_{jk}^{-8}$ kernel is dominated by the closest cross-group pair,
and we therefore limit the group size by $n_{\max}$.
\begin{remark}
\label{rem:cliques}
We also compare isolated cliques of three to five spacecraft and an
octahedron under random commands, where the ratio grows with the
clique size, whereas the octahedron gains less than the 5-clique
despite more edges. The gain therefore depends on the directions of
the edges of a spacecraft, not on their number.
\end{remark}
\begin{table}[tb!]
\centering
\caption{Power ratios $J(Z(2))/J(Z(n))$ on isolated cliques and
$W(Z(2))/W(\mathrm{method})$.}
\label{fig:density_bench}
\label{fig:capsweep_peakshift}
\begin{tabular}{|l|c|c||l|c|c|}\hline
Clique & Med. & Max & Mission & Med. & Max \\ \hline
3-clique & $1.75$ & $2.9$ & Inv.\ $K{=}10^{2}$ & $0.79$ & $0.94$ \\ \hline
4-clique & $2.42$ & $5.1$ & Inv.\ $K{=}10^{5}$ & $0.93$ & $0.94$ \\ \hline
5-clique & $4.02$ & $9.0$ & ODA $n{=}2$ & $1.00$ & $1.00$ \\ \hline
Octahedron & $3.19$ & $6.3$ & Grp.\ $n{=}4$ & $1.39$ & $15.9$ \\ \hline
\end{tabular}
\end{table}
\subsection{Stage 3: Swarm Kinematics Controller for Groups}
\label{sec:skc}
\label{sec:momentum_generalization}
The third stage constructs a decentralized controller that converges to target states while conserving angular momentum. We first generalize the momentum constraint \cite{takahashi2022kinematics} to the full state of $N$ spacecraft and $N_{\mathrm w}$ wheels:
\begin{equation*}
\mathcal A(N,N_{\mathrm w};\bm q)\bm\zeta
=
\left[\begin{smallmatrix}
m_1E_3 & m_1[{}^i\bm r_1]_\times^\top\\
\vdots & \vdots\\
m_NE_3 & m_N[{}^i\bm r_N]_\times^\top\\
\bm0 & (C^{I/B_1}\bar I_1)^\top\\
\vdots & \vdots\\
\bm0 & (C^{I/B_N}\bar I_N)^\top\\
\bm0 & C^{I/B_{j_1}\top}\\
\vdots & \vdots\\
\bm0 & C^{I/B_{j_{N_{\mathrm w}}}\top}
\end{smallmatrix}\right]^{\!\top}\!\bm\zeta
=\begin{bmatrix}\bm P_n^{i}\\\bm L_n^{i}\end{bmatrix},\\
\end{equation*}
where $[\cdot]_\times$ is the skew operator and $E_3$ the identity. We define the local constraint with $n$ members and $m$ wheels:
\begin{equation}
A(n,m;\bm q)\bm\zeta
=
\left[\begin{smallmatrix}
m_1E_3 & m_1[{}^i\bm r_1]_\times^\top\\
\vdots & \vdots\\
m_nE_3 & m_n[{}^i\bm r_n]_\times^\top\\
\bm0 & (C^{I/B_1}\bar I_1)^\top\\
\vdots & \vdots\\
\bm0 & (C^{I/B_n}\bar I_n)^\top\\
\bm0 & C^{I/B_{j_1}\top}\\
\vdots & \vdots\\
\bm0 & C^{I/B_{j_m}\top}
\end{smallmatrix}\right]^{\!\top}\!\bm\zeta=
\begin{bmatrix}\bm P_N^{i}\\\bm L_N^{i}\end{bmatrix}.
\label{eq:momentum_constraint_rows}
\end{equation}
Each local group computes $S(n_\ell,m_\ell)$ from $A(n,m;\bm q)$ and applies the admissible input of
Lemma~\ref{NODA_theorem_experimental_controller} \cite{takahashi2022kinematics}:
\begin{equation}
\begin{aligned}
\bm u_\ell&=B(n_\ell,m_\ell)^{-1}M(n_\ell,m_\ell)S(n_\ell,m_\ell)\mathfrak u_\ell,\\
\mathfrak u_\ell
&=\dot{\bm v}_{\mathrm r,\ell}
+\overline M_\ell^{-1}
\bigl(
\overline C_\ell\bm v_{\mathrm r,\ell}
-K_{1,\ell}\,\delta\bm q_{s,\ell}
-K_{2,\ell}\,\delta\bm v_\ell
\bigr),
\end{aligned}
\label{eq:wheelequipped_reduced_command}
\end{equation}
where $\delta\bm q_{s,\ell},\delta\bm v_\ell$ are the tracking errors,
$K_{1,\ell},K_{2,\ell}\succ0$, and $\overline M_\ell=S^\top MS$ and
$\overline C_\ell=S^\top(M\dot S+CS)$ are evaluated at
$(n_\ell,m_\ell)$ with $C$ the Coriolis matrix of
\cite{takahashi2022kinematics}. 
\begin{lemma}
\label{lemma:decentral_null_space}
Consider the member embeddings
$K_\ell\in\mathbb R^{(6N+3N_{\mathrm w})\times(6n_\ell+3m_\ell)}$, the
selection matrices with orthonormal columns that place the
coordinates of the members of $\mathfrak g_\ell$ at their positions in
the swarm coordinates, stacked into $\mathcal S=[K_1S(n_1,m_1)\ \cdots\ K_{N_g}S(n_{N_g},m_{N_g})]$. Then, $\mathcal A(N,N_{\mathrm w};\bm q)\mathcal S=\bm0$, and for a
disjoint partition $\mathcal S$ has
$6N+3N_{\mathrm w}-6N_g$ columns and leaves $6(N_g-1)$ directions of
$\ker\mathcal A(N,N_{\mathrm w};\bm q)$ unspanned.
\end{lemma}
\begin{proof}
The embedding gives
$\mathcal A(N,N_{\mathrm w};\bm q)K_\ell=A(n_\ell,m_\ell;\bm q)$, and
applying it block by block,
\begin{equation*}
\begin{aligned}
\mathcal A(N,N_{\mathrm w};\bm q)\mathcal S
&=\bigl[\,\mathcal A(N,N_{\mathrm w};\bm q)K_1S(n_1,m_1)\ \cdots\,\bigr]\\
&=\bigl[\,A(n_1,m_1;\bm q)S(n_1,m_1)\ \cdots\,\bigr]=\bm0,
\end{aligned}
\end{equation*}
since $A(n_\ell,m_\ell;\bm q)S(n_\ell,m_\ell)=\bm0$.
For a disjoint partition $K_\ell^\top K_q=\bm0$ for $\ell\neq q$, and
$\mathcal S$ has full column rank, $6(N_g-1)$ below
$\dim\ker\mathcal A(N,N_{\mathrm w};\bm q)=6N+3N_{\mathrm w}-6$.
\end{proof}
Lemma~\ref{lemma:decentral_null_space} derives the averaged dynamics of the swarm with the constraint force of the momentum conservation:
\begin{equation*}
M\dot{\bm\zeta}+C\bm\zeta=B\sum\nolimits_\ell K_\ell\bm u_\ell
-\mathcal A(N,N_{\mathrm w};\bm q)^\top\bm\eta,
\end{equation*}
where $M$, $C$, and $B$ are block diagonal over the spacecraft and
wheels. On the zero-momentum leaf of every group, $\bm\zeta=\mathcal S\bm v$ with $\bm v$ stacking the reduced velocities $\bm v_\ell$. Multiplying the dynamics from the
left by $\mathcal S^\top$ removes the constraint force, since
$\mathcal S^\top\mathcal A^\top=\bm0$, and the block-diagonal
structure with $K_\ell^\top K_q=0$ for $\ell\neq q$ decouples the
result over the groups,
$\overline M_\ell\dot{\bm v}_\ell+\overline C_\ell\bm v_\ell
=\overline M_\ell\mathfrak u_\ell$. Substituting
\eqref{eq:wheelequipped_reduced_command} gives the error dynamics of
each group,
\begin{equation}
\overline M_\ell\delta\dot{\bm v}_\ell+\overline C_\ell\delta\bm v_\ell
=-K_{1,\ell}\delta\bm q_{s,\ell}-K_{2,\ell}\delta\bm v_\ell,
\label{eq:group_error_dynamics}
\end{equation}
where $\delta\bm q_{s,\ell}$ stacks the position errors
$\delta\bm r_j$ and the MRP errors $\delta\bm\sigma_j$ of the members
and $K_{1,\ell}=\operatorname{diag}(k_pE,k_\sigma E,\bm0)$.
\begin{theorem}
\label{NODA_theorem_overlapping_groups}
Assume that the commanded graph is connected, that the groups of each carrier form a disjoint partition, and that each group follows
\eqref{eq:group_error_dynamics}. Then
all positions and attitudes converge to the reference, and the wheel
momenta to the uniform share
\begin{equation}
\lim_{t\to\infty}\bm h_j^i
=
(\bm L^i-\bm L_d^i)/N_{\mathrm w},
\qquad
j=1,\ldots,N_{\mathrm w}.
\label{NODA_uniform_RW_distribution}
\end{equation}

Moreover, if the allocated
amplitudes are Lipschitz continuous in the state, the state errors
remain bounded, shrinking with the carrier frequency $\omega_0$ and
growing with the cross-group disturbance in \eqref{eq:MIG_AC_cut_penalty}.
\end{theorem}
\begin{proof}
Memberships overlap across carriers through \eqref{NODA_multi_frequency_dipole_moment} and superpose, and the averaged interaction \eqref{NODA_unified_averaged_interaction} is additive over carriers. Then, a spacecraft in several groups receives the sum of their commands, its wheels absorb the sum of their shares, and the sums below run over spacecraft. Let $V=\sum_\ell V_\ell$ with
$V_\ell=\frac{\delta\bm v_\ell^\top\overline M_\ell\delta\bm v_\ell}{2}
+\tfrac{k_p}{2}\sum_{j\in\mathfrak g_\ell}\|\delta\bm r_j\|^2
+2k_\sigma\sum_{j\in\mathfrak g_\ell}\ln(1+\|\delta\bm\sigma_j\|^2)$
\cite{takahashi2022kinematics}. Along \eqref{eq:group_error_dynamics},
the skew symmetry of $\dot{\overline M}_\ell-2\overline C_\ell$ and the
attitude kinematics cancel the $K_{1,\ell}$ terms, and between exchanges
$\dot V=-\sum_\ell\delta\bm v_\ell^\top K_{2,\ell}\delta\bm v_\ell\leq0$.
The wheel exchange is the implicit consensus step
$(E+k_\xi T_u\,\mathcal L_h\otimes E_3)\,\bm h^{i+}=\bm h^i$ with the Laplacian of the wheel-coupling graph $\mathcal L_h$. For any
$k_\xi T_u>0$, this 
does not increase $V$ and it becomes the consensus flow $\dot{\bm h}^i=-k_\xi(\mathcal L_h\otimes E_3)\bm h^i$, for which $\frac{\mathrm d}{\mathrm dt}\,\frac{\|\bm h^i\|^2}{2}
=-k_\xi\,\bm h^{i\top}(\mathcal L_h\otimes E_3)\bm h^i\leq 0$. Then, the connected graph makes the largest invariant set the zero-velocity-error set
with position and wheel-momentum consensus. LaSalle's theorem \cite{khalil2002nonlinear} gives
convergence, and the conserved total momentum fixes the wheel momenta
at \eqref{NODA_uniform_RW_distribution}. The $6(N_g-1)$ directions of Lemma~\ref{lemma:decentral_null_space} that a disjoint partition leaves unspanned carry the relative momenta of the groups.

For the AC-driven swarm, the interaction \eqref{Nonlinear_NODA_MMH_multifrequency_complete_interaction} is smooth on separated configurations and the Lipschitz amplitudes give a unique solution. With all carriers integer multiples of $\omega_0$, the closed loop is $T$-periodic with \eqref{NODA_unified_averaged_interaction} as its period average, and first-order averaging \cite{sanders2007averaging} keeps the true trajectories within $O(1/\omega_0)$ of the averaged ones over finite time. Since the equilibrium of \eqref{eq:group_error_dynamics} is locally exponentially stable on the level sets of the conserved momenta, a strict converse Lyapunov function exists, and input-to-state stability against the cross-group residual bounded by \eqref{eq:MIG_AC_cut_penalty} extends the error bound to
all time. The ODA solves a sequence of convex programs with a unique minimizer on separated configurations, whose solution is Lipschitz in the command and the geometry.
\end{proof}
\subsection{Calculation Mitigation with Frozen Approximation}
\label{sec:frozen}
The multiple carriers in
\eqref{Nonlinear_NODA_MMH_multifrequency_complete_interaction} require
small steps and are computationally expensive, motivating the
approximation. 
\begin{proposition}
\label{thm:frozen}
Consider every interaction matrix $Q_{j\leftarrow k}$ of \eqref{EMgrouping_dipolemodel} fixed at the
start of each interval and the closed loop integrated with the amplitudes fixed over
update intervals $T_u$ of whole common periods. Then, on a connected graph, its steady-state
position and attitude errors with the actual trajectory are
bounded by
\begin{equation}
\|\delta P\|\leq\frac{4k_{\max}q_{\max}\bar F}{k_p\lambda_2},
\qquad
\|\delta\bm\sigma_j\|\leq\frac{3q_{\max}\bar T_j}{k_\sigma},
\label{eq:frozen_bounds}
\end{equation}
where $q_{\max}=\max_ev_eT_u/d_e$ is the largest relative displacement
of an edge, $v_e$ bounds the relative speed
over the edge length $d_e$, $\bar F$ bounds the force of every magnetic edge,
$\bar T_j=\sum_{e\ni j}\|\bm\tau_{j,e}\|$ is the torque of the magnetic edges of spacecraft $j$, and $k_{\max}$ is the largest number of magnetic partners of a spacecraft, which exceeds $2k$ when more than $k$ spacecraft select the same neighbor.
\end{proposition}
\begin{proof}
The disturbances $\bm d_{j\leftarrow k}^{2\omega_p}$ and
$\bm d_{j\leftarrow k}^{pq}$ of
\eqref{Nonlinear_NODA_MMH_multifrequency_complete_interaction} have
zero mean over whole common periods, and the only remaining error
comes from the motion of $Q_{j\leftarrow k}$, at most $v_eT_u$ on an
edge $e=(j,k)$. For the force, the $d^{-4}$ law gives
$\|\partial Q/\partial\bm r\|\leq4\|Q\|/d$ and bounds the error of the
edge by $\|\delta\bm f_e\|\leq4q_e\|\bm F_e\|$ with $q_e=v_eT_u/d_e$.
This error acts as $+\delta\bm f_e$ on spacecraft $j$ and
$-\delta\bm f_e$ on spacecraft $k$, and the per-spacecraft errors
$\delta\bm F$ therefore sum to zero. Summing at most $k_{\max}$
incident edges gives $\|\delta\bm F_j\|\leq4k_{\max}q_{\max}\bar F$. In
steady state the position errors of \eqref{centralized_controller}
satisfy $k_pL\,\delta P=\delta\bm F$, the zero sum places
$\delta\bm F$ orthogonal to the kernel of $L$ on a connected graph,
and the smallest nonzero eigenvalue $\lambda_2$ of $L$ gives the first
bound. Likewise, the $d^{-3}$ law bounds the torque error of edge $e$
on spacecraft $j$ by $3q_e\|\bm\tau_{j,e}\|$. These errors are not
equal and opposite on the two members, and they act on each body
directly with $\|\delta\bm T_j\|\leq3q_{\max}\bar T_j$, and the attitude command
settles at $k_\sigma\delta\bm\sigma_j=\delta\bm T_j$.
\end{proof}
\section{Application to Spacecraft Swarm Reconfiguration}
\label{sec:evaluation}
\label{sec:eval_proposed}
\label{sec:full_AC_1000sat}
This section verifies the proposed framework on two illustrative
examples. Note that the simulation integrates the full dynamics by fourth-order
Runge--Kutta and re-forms the groups at every frame. 
The amplitudes are updated every common period $T=2\pi/\omega_0=0.53$~s
of the base carrier, and the step $\Delta t=T/(N_{\mathrm{sub}}\max_p\nu_p)$ resolves the highest carrier with $N_{\mathrm{sub}}$ sub-steps.
\label{sec:V_channels}
\subsection{Tower Reconfiguration for Framework Validation}
\label{sec:V_tuning}
Fig.~\ref{fig:tower_formation} shows the terminal formation of the
$1000$ spacecraft. Each spacecraft has $m_{\mathrm s}=20$~kg and an
isotropic inertia $J_{\mathrm s}=20$~kg\,m$^2$, and the dipoles are
uncapped, an idealization that ignores the actuator limits of a
flight unit. The loop delay is the hold delay $\tau_0/2$ alone with
$\tau_{\mathrm{comm}}=0$, the highest carrier is the $1521$st harmonic
of $\omega_0$ at $7.2$~kHz, and $N_{\mathrm{sub}}=20$.
\begin{figure}[tb!]
\centering
\subfloat[Oblique view.]{
\includegraphics[width=0.41\linewidth]
{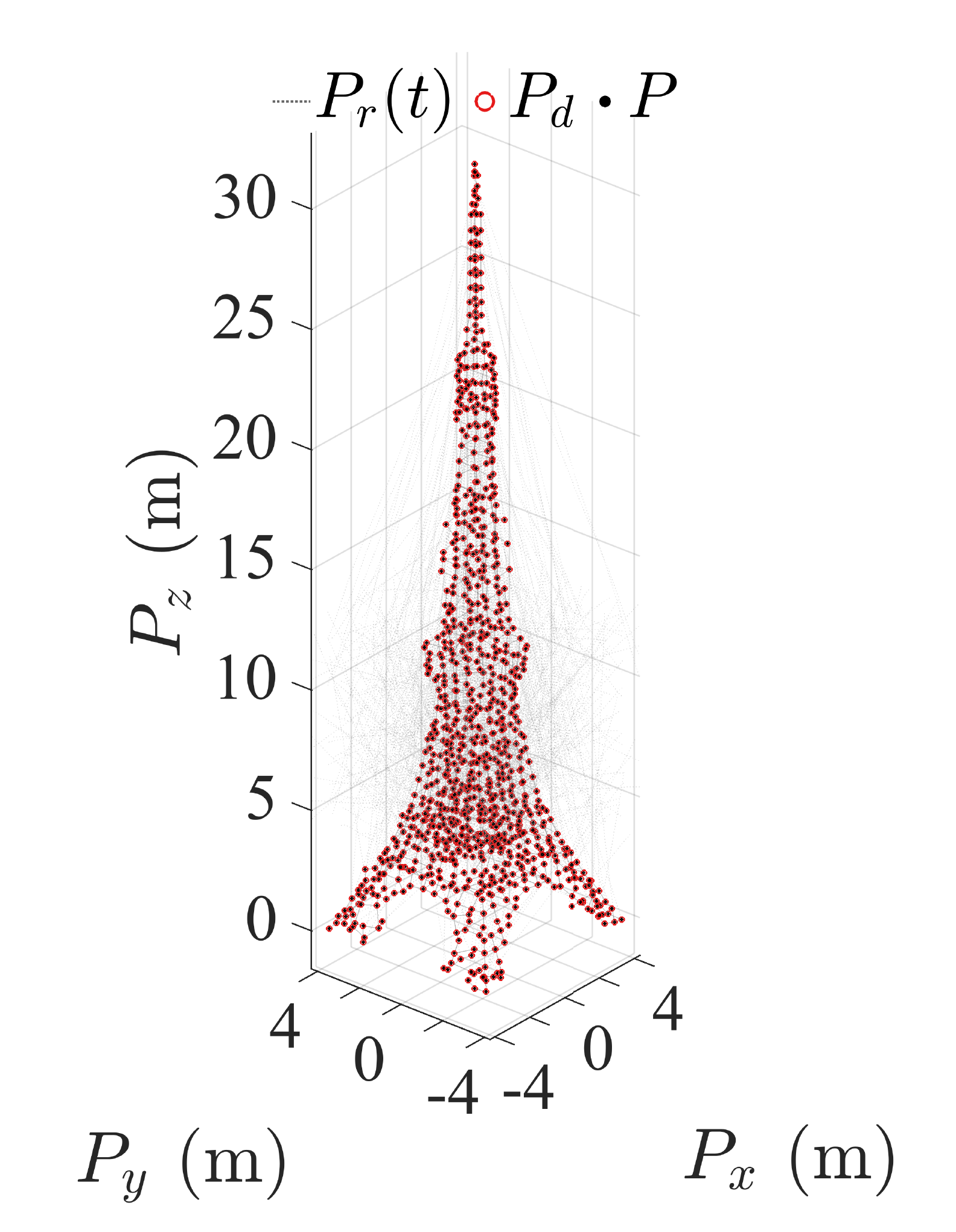}
\label{fig:tower_oblique}}
\subfloat[Top view.]{
\includegraphics[width=0.51\linewidth]
{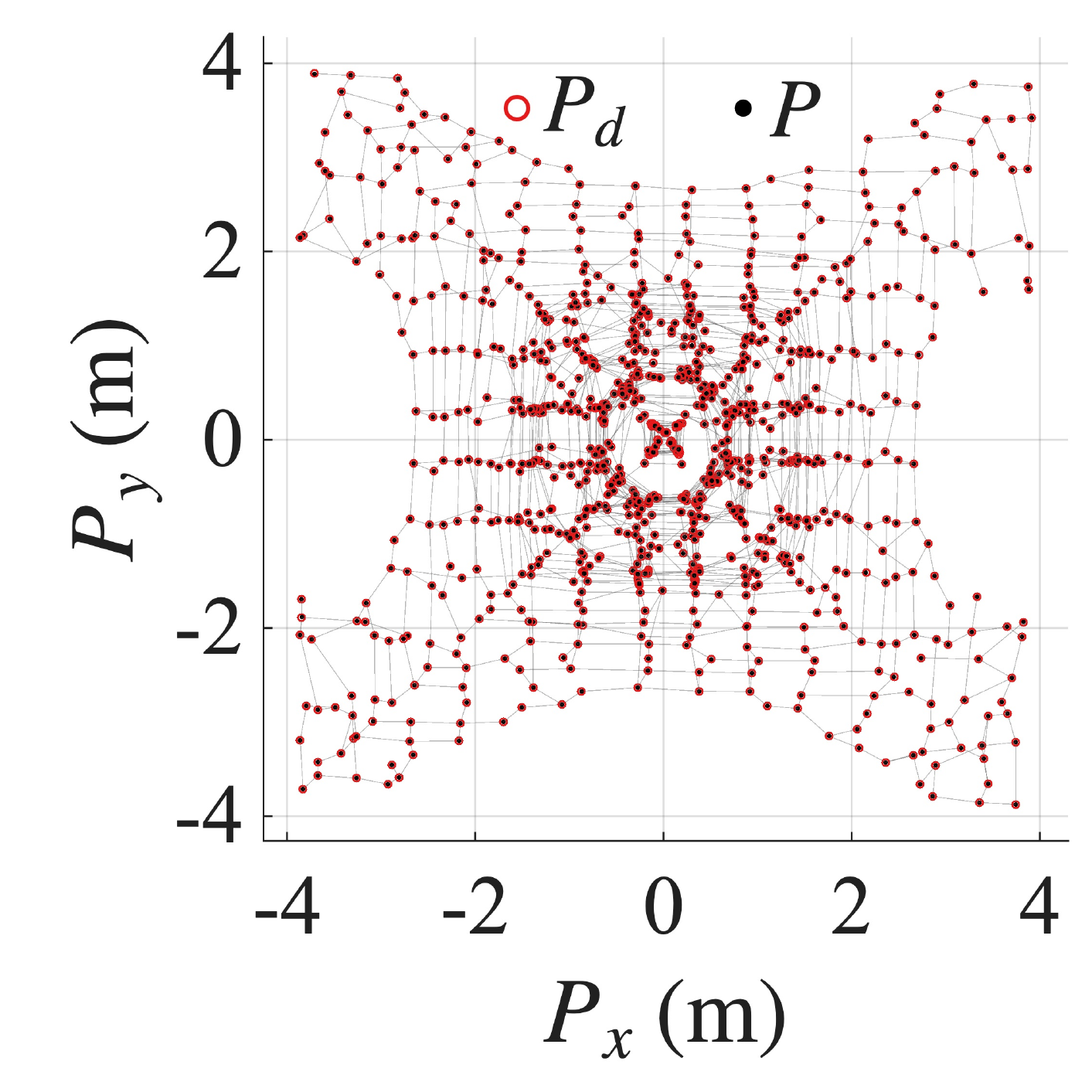}
\label{fig:tower_top}}
\caption{The terminal formation representing the Tokyo Tower (Tokyo Digital Twin Project, CC-BY).}
\label{fig:tower_formation}
\end{figure}
\subsubsection{MIG-Tuning of Group-Size Limit}
\label{sec:V_kinematics}
We sweep the group-size limit $n$ from two to seven under the fixed commanded graph in Fig.~\ref{fig:fullAC_1000sat_summary}, in which every spacecraft commands its eight nearest neighbors. The controller gains of \eqref{centralized_controller} are set to $(k_p,k_v,k_\sigma,k_\omega)=(0.1,0.4,1.2,1.2)$ at the $0.53$~s hold. The dipole power decreases monotonically with $n$ and $n{=}7$ operates at $40\%$ of the power of $n{=}2$ in Fig.~\ref{fig:fullAC_1000sat_summary}(b). MIG partitions that graph, and the peak number of carriers in Fig.~\ref{fig:fullAC_1000sat_summary}(f) falls from ${\sim}3400$ at $n{=}2$ to ${\sim}840$ at $n{=}7$. Against the best-of-$K$ inverse allocation with $K{=}1000$, the ODA lowers the total power in all $57$ frames of the reconfiguration at every $n$ (Table~\ref{tab:n234_ripple_gain}). The time-averaging error, the root-mean-square AC disturbance force of \eqref{Nonlinear_NODA_MMH_multifrequency_complete_interaction}, is not minimized and falls less uniformly over the $5.7\times10^{4}$ spacecraft-frames, since the intra-group bound of \cite{takahashi2024neural} ties the disturbance to the power without ordering two allocations of the same pair. The tracking improves monotonically with $n$, and the position error falls from $0.47$ to $0.09$ times the median spacing of $0.38$~m and from $2.6$ to $0.6$ for the worst spacecraft in Fig.~\ref{fig:fullAC_1000sat_summary}(c).
\begin{table}[tb!]
\centering
\caption{Gain over the best-of-$K$ inverse allocation ($K{=}1000$) in the closed-loop sweep, total power per frame and time-averaging error per spacecraft.}
\label{tab:n234_ripple_gain}
\begin{tabular}{|l|c|c|c|}
\hline
& $n{=}2$ & $n{=}3$ & $n{=}4$ \\ \hline
Power gain (median) & $1.15\times$ & $1.34\times$ & $1.58\times$ \\ \hline
Power gain (range) & $1.05$--$1.94\times$ & $1.23$--$2.29\times$ & $1.39$--$2.78\times$ \\ \hline
Error gain (median) & $1.04\times$ & $1.28\times$ & $1.54\times$ \\ \hline
Error gain (10--90\%) & $0.40$--$2.50\times$ & $0.58$--$2.94\times$ & $0.76$--$3.67\times$ \\ \hline
Lower error & $54\%$ & $73\%$ & $83\%$ \\ \hline
\end{tabular}
\end{table}
The wheel momenta converge toward \eqref{NODA_uniform_RW_distribution} and the power falls by a further factor of two, as the loop only holds the formation in Fig.~\ref{fig:fullAC_1000sat_summary}(e). 
\begin{figure*}[t!]
\centering
\begin{minipage}[t]{0.485\textwidth}
\vspace{0pt}
\centering
\subfloat[Legend.]{\includegraphics[width=0.49\linewidth]{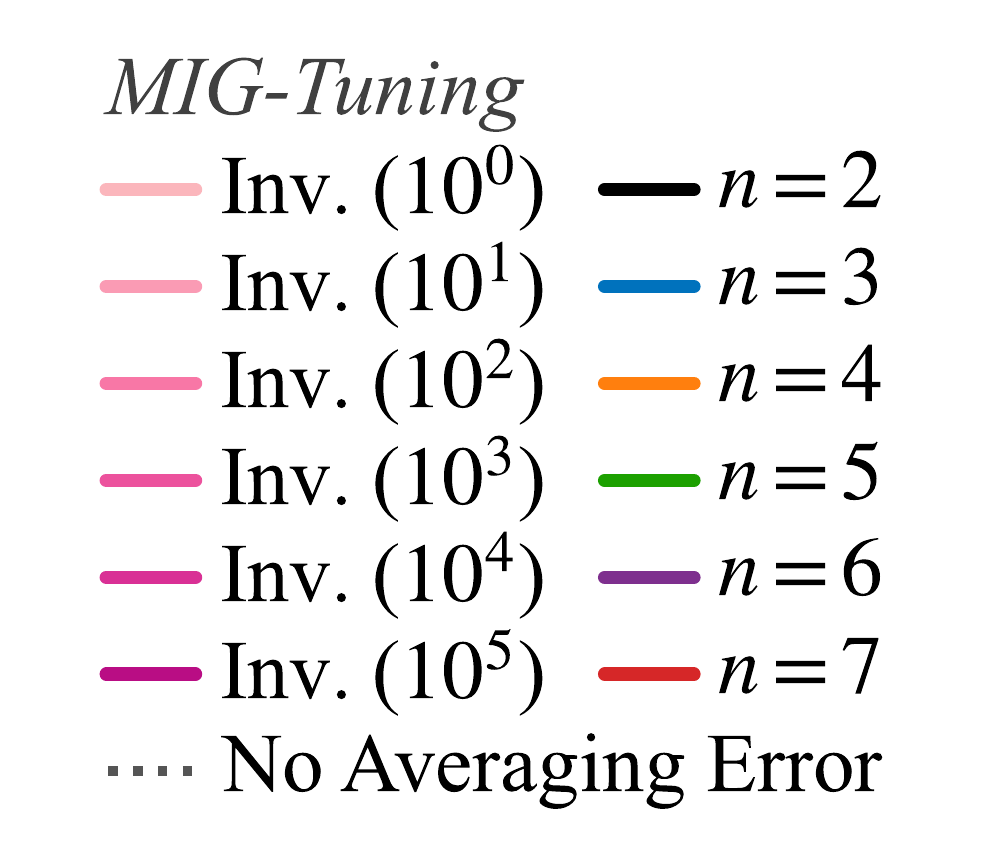}}\hfill
\subfloat[Energy ratio.]{\includegraphics[width=0.49\linewidth]{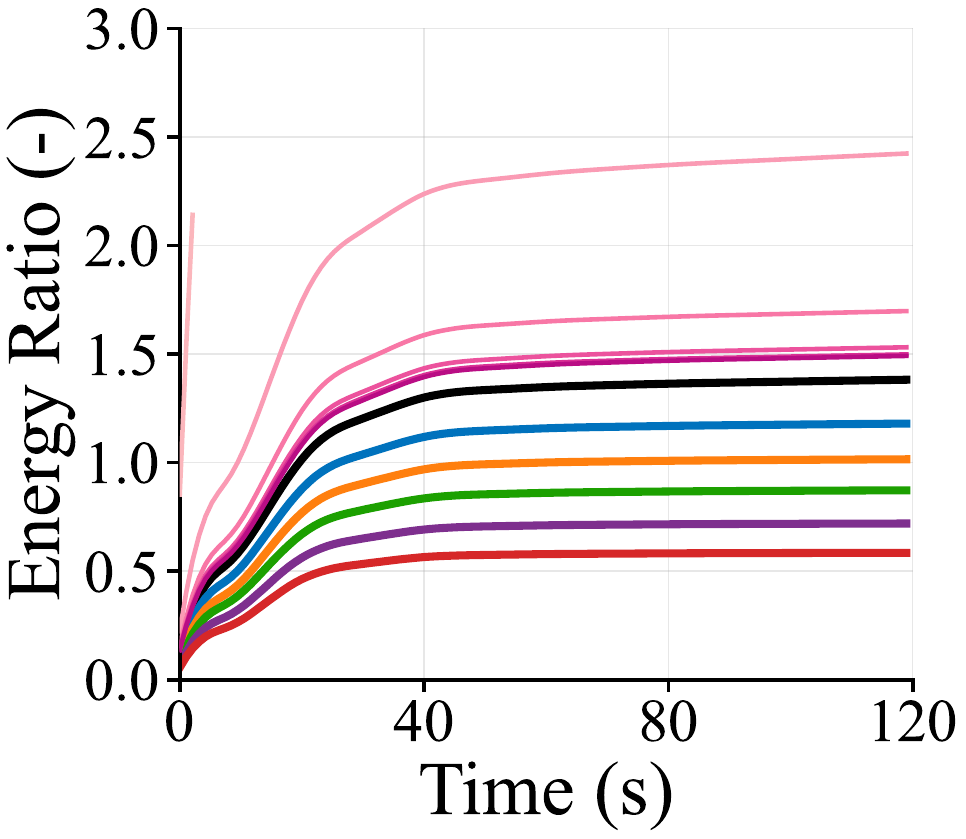}}\\
\subfloat[Position.]{\includegraphics[width=0.49\linewidth]{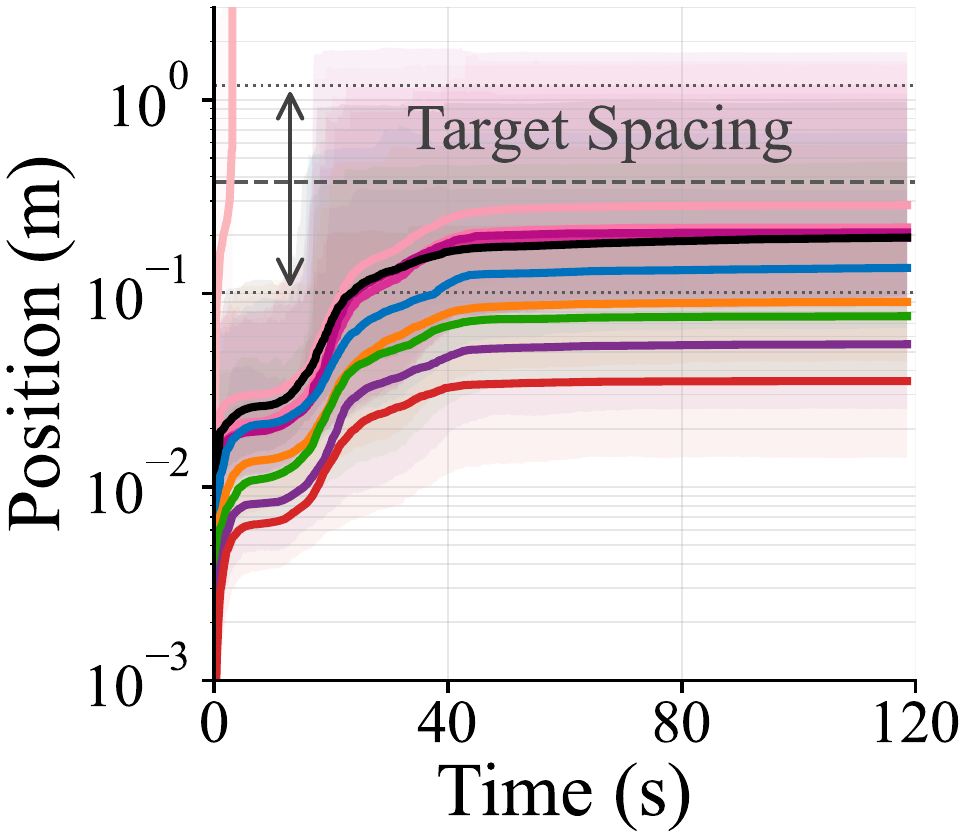}}\hfill
\subfloat[Attitude.]{\includegraphics[width=0.49\linewidth]{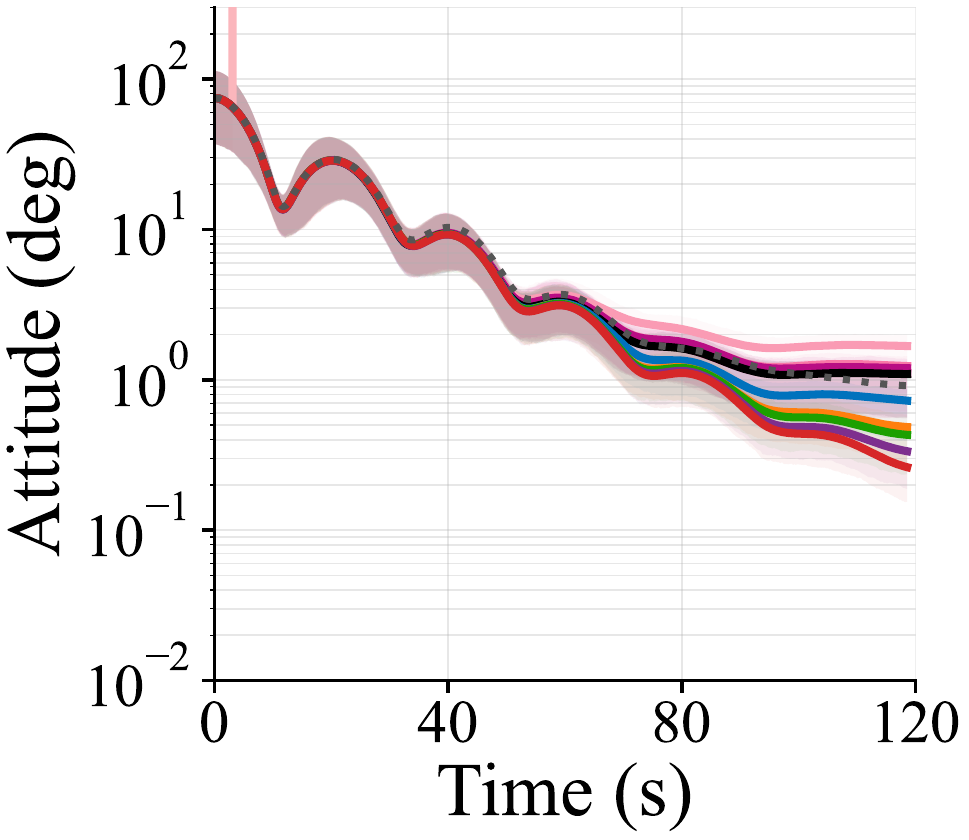}}\\
\subfloat[RW momentum.]{\includegraphics[width=0.49\linewidth]{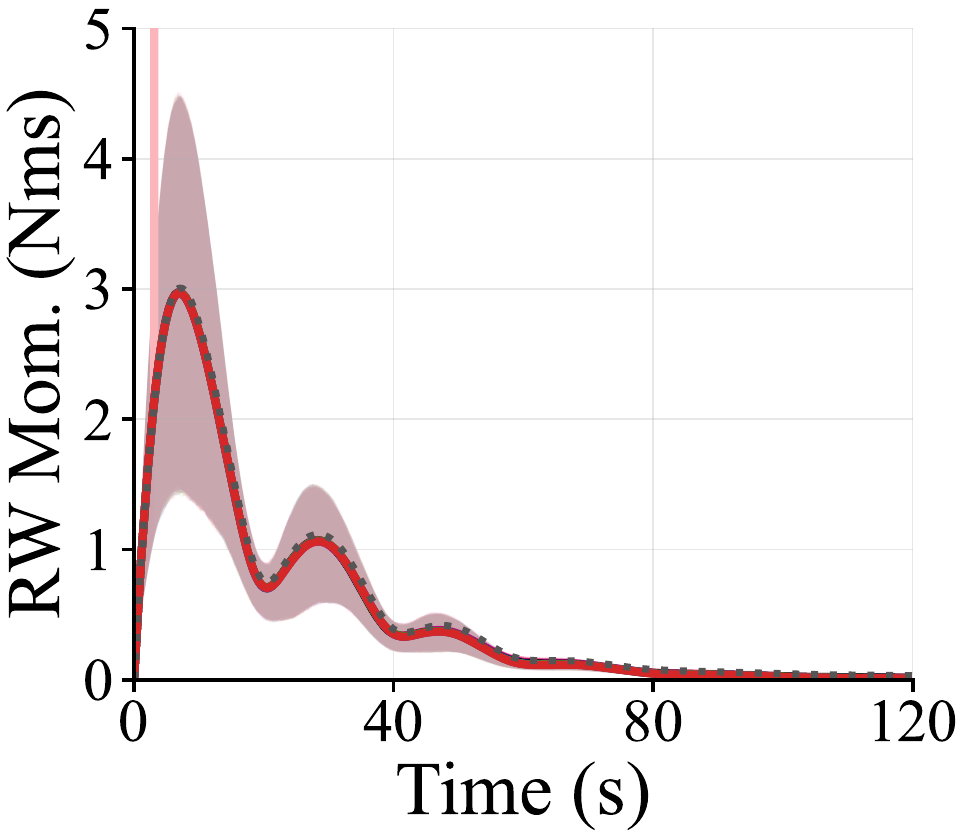}}\hfill
\subfloat[Distinct AC carriers.]{\includegraphics[width=0.49\linewidth]{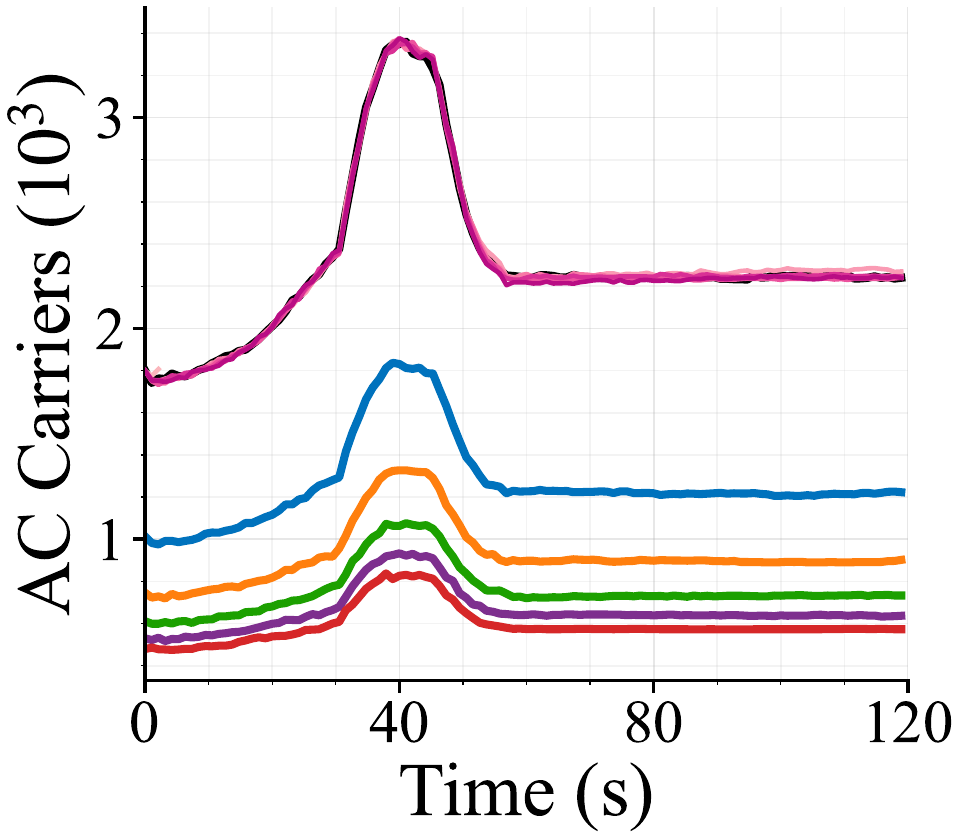}}
\caption{Reconfiguration for MIG-tuning in Section~V-A: closed loop (solid) against the averaged model (dotted).}
\label{fig:fullAC_1000sat_summary}
\label{fig:ac_frequency_schedule}
\label{Nonlinear_NODA_MMH_closed_loop_validation}
\end{minipage}\hfill
\begin{minipage}[t]{0.485\textwidth}
\vspace{0pt}
\centering
\subfloat[Legend.]{\includegraphics[width=0.49\linewidth]{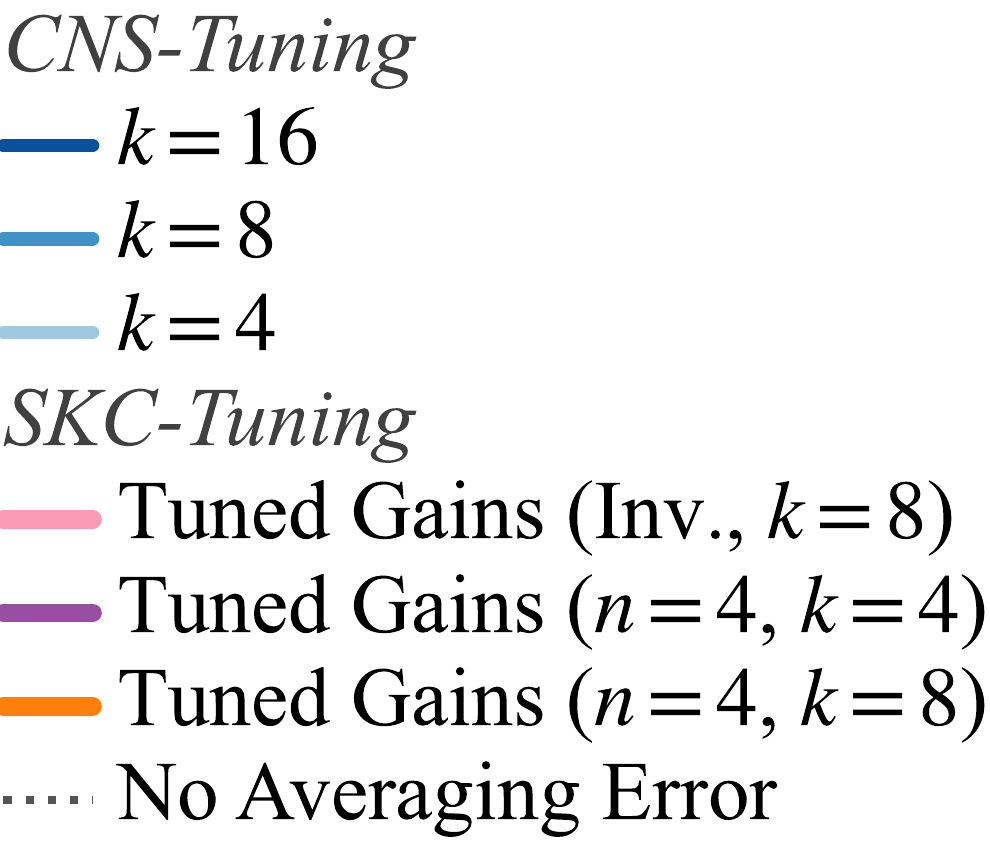}}\hfill
\subfloat[Energy ratio.]{\includegraphics[width=0.49\linewidth]{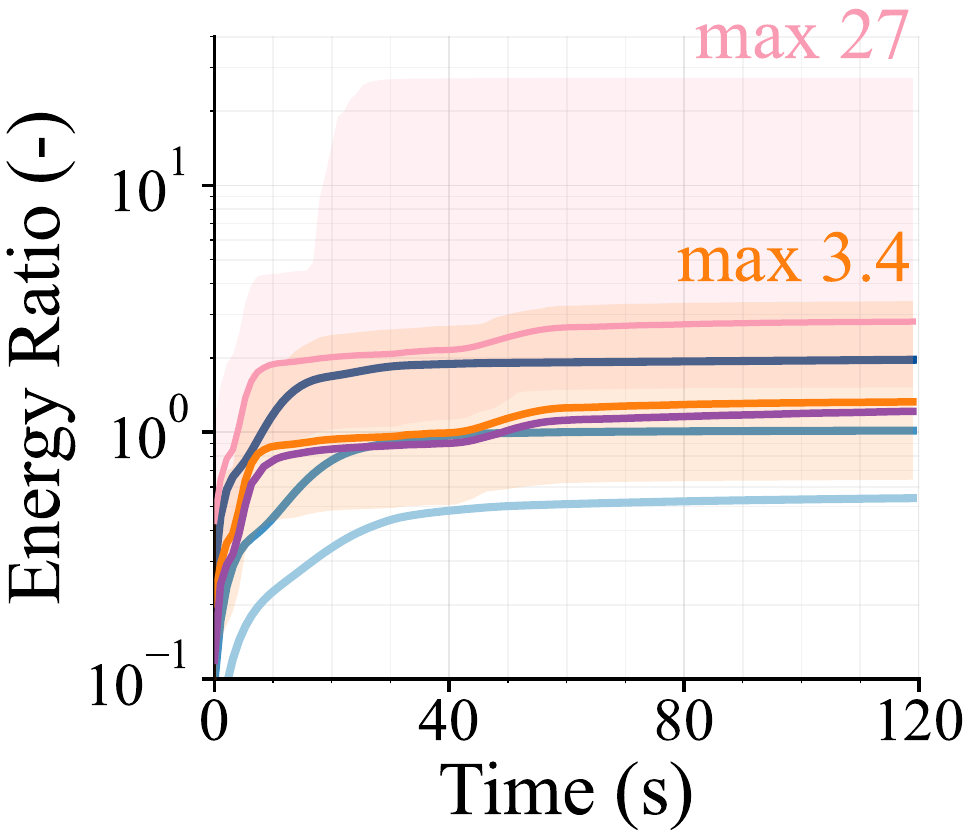}}\\
\subfloat[Position.]{\includegraphics[width=0.49\linewidth]{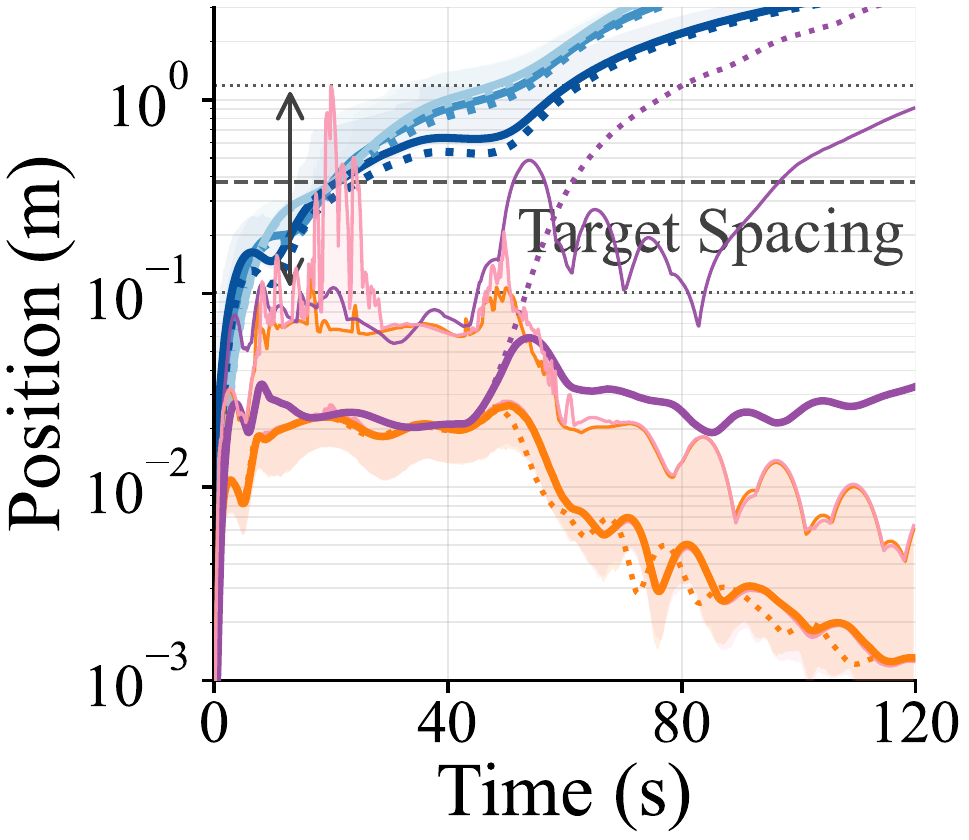}}\hfill
\subfloat[Attitude.]{\includegraphics[width=0.49\linewidth]{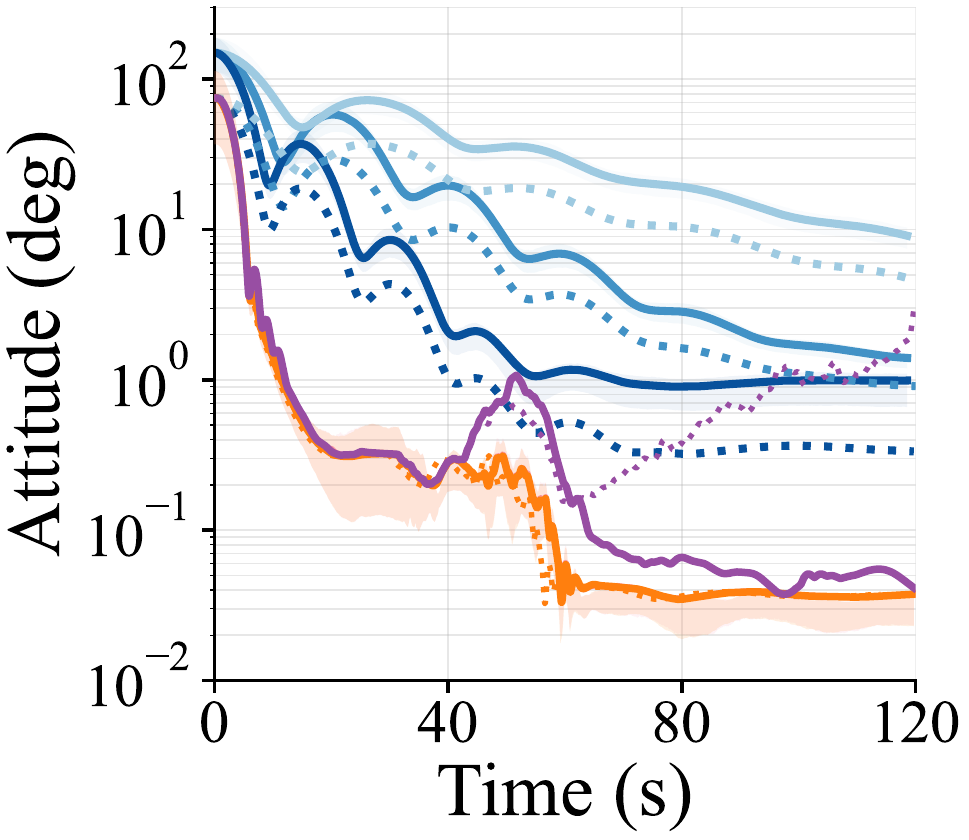}}\\
\subfloat[RW momentum.]{\includegraphics[width=0.49\linewidth]{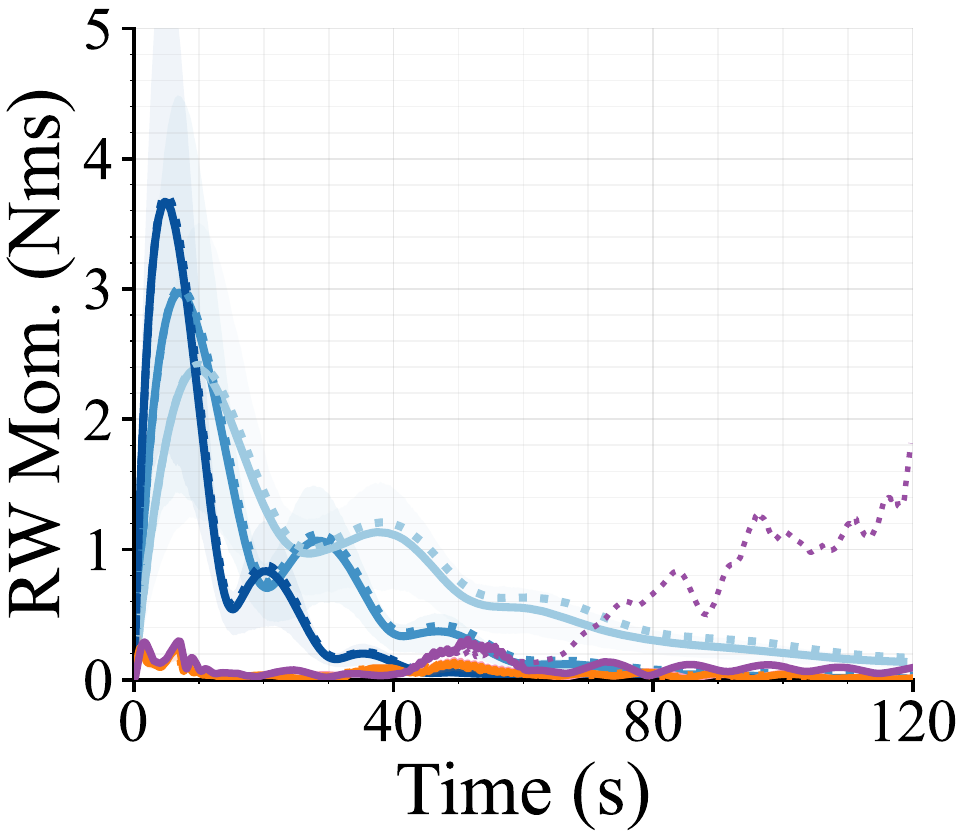}}\hfill
\subfloat[Distinct AC carriers.]{\includegraphics[width=0.49\linewidth]{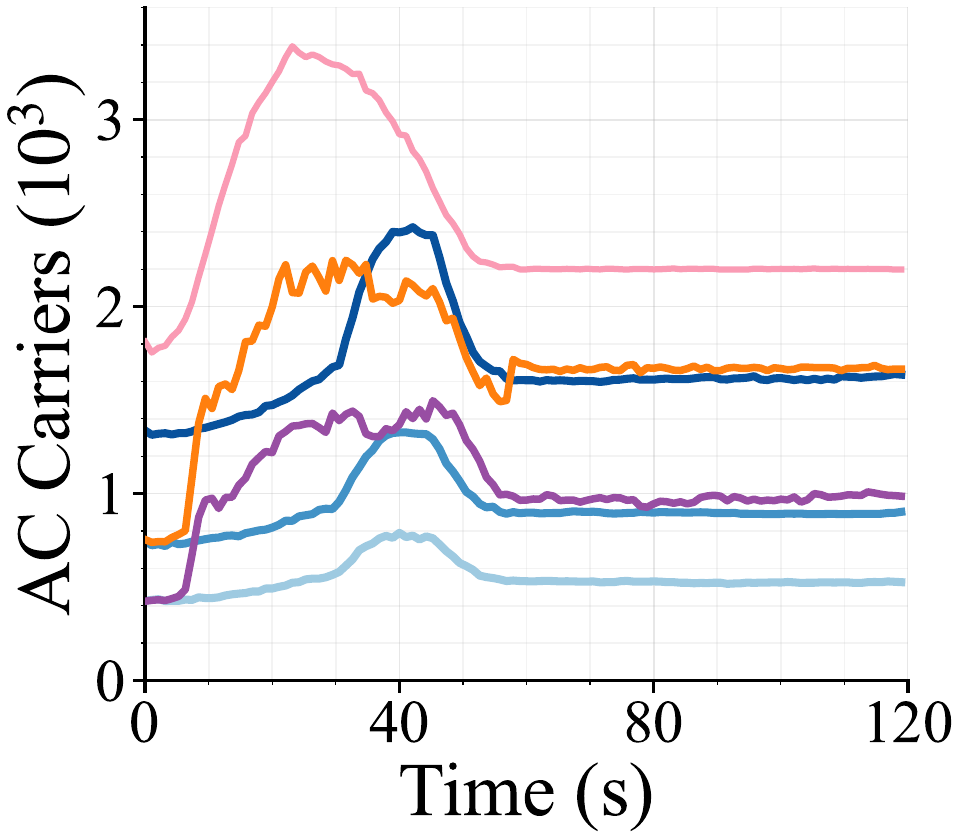}}
\caption{Reconfiguration for CNS- and SKC-tuning in Section~V-A, with the conventions of Fig.~\ref{fig:fullAC_1000sat_summary}.}
\label{fig:cns_ablation}
\label{fig:codesign_progress}
\end{minipage}
\end{figure*}
\subsubsection{CNS-Tuning of Neighbor Count}
\label{sec:V_ablation}
Criterion~2 of Section~IV-A floors the neighbor count first. Over the $601$ geometries of the tower
reconfiguration, at the consensus stiffness calibrated so that the worst geometry of $k{=}20$ diverges at $35$~ms, the worst
$\tau^{\mathrm{div}}$ rises monotonically with $k$, from $12.5$~ms at
$k{=}8$ and $27.5$~ms at $k{=}16$ to $42.5$~ms at $k{=}24$ (Fig.~\ref{fig:model_norm_blocks}(c)); a sparser graph tolerates
only a shorter hold at the same consensus rate.
With the grouping fixed at $n{=}4$, we sweep the symmetrized $k$-NN graph from $k{=}4$ to $k{=}16$ in Fig.~\ref{fig:cns_ablation}. A larger $k$ lengthens the mean effective distance $\langle d_{\mathrm{eff}}\rangle$ of \eqref{eq:deff_pooling}, and the power follows its fourth power within $9\%$ in Fig.~\ref{fig:cns_ablation}(b). 
A larger $k$ also raises $\lambda_2$ at the cost of this power and of more carriers. The worst tracking error among all spacecraft at $t{=}60$~s is $1.7$~m at $k{=}16$, whereas it is $2.3$~m at $k{=}4$, and by $t{=}120$~s the two differ twofold, $4.3$ against $8.3$~m.
Even so, the consensus is slow at every $k$, and the $k$-NN graph gives $\lambda_2=0.0092$ at $k{=}8$ in the terminal formation. We therefore add $150$ information edges of Criterion~3, $3\%$ of the edge list, drawn once
uniformly at random among the pairs more than $8$~m apart in the
terminal formation, which raise $\lambda_2$ to $0.27$.
\subsubsection{SKC-Tuning of Controller Gains}
\label{sec:V_codesign}
We finally tune the gains of \eqref{centralized_controller} under the grouping fixed at $n{=}4$ and the graph fixed at $k{=}8$. Since MIG of $n{=}4$ cuts the peak carrier count from ${\sim}3400$ to $1327$, \eqref{eq:delay_bound} admits a $2.6\times$ shorter hold, from $0.53$~s to $0.21$~s, and higher gains. The gains are tuned on the linearized averaged model from an LQR initialization to minimize the error within the mission time, and the selected set $(k_p,k_v,k_\sigma,k_\omega)=(5.0,4.0,56,5.0)$ is $4$--$50$ times the gains of the sweeps.
Against the mean target spacing $\bar d(t)$, every lane of Sections~V-A-1 and V-A-2 peaks at
$1.2$--$5.0\,\bar d$, whereas the tuned loop peaks at $0.24\,\bar d$
and its tracking error settles to $6$~mm by $120$~s. The same gains diverge on the sparser $k{=}4$ graph. In the averaged model the worst-spacecraft error reaches $3.9$~m at $t=60$~s and $89$~m at $t=120$~s, against $0.02$~m at $t=60$~s on $k{=}8$.
\begin{remark}
\label{rem:hold}
The gain limit scales inversely with the amplitude hold. In the
averaged model with $k{=}8$, the position loop diverges at
$k_v=6.4$, $16$, and $26$ ($k_p=k_v/4$) for holds $\tau_0=0.53$,
$0.21$, and $0.13$~s, respectively.
With the hold delay $\tau=\tau_0/2$, the bound
\eqref{eq:delay_bound} reads $(k_v/m_{\mathrm s})\lambda_N\tau_0<\pi$,
and the measured divergence boundary
$(k_v/m_{\mathrm s})\lambda_N\tau_0\approx2.7$--$2.8$ follows it at all
three holds.
\end{remark}
\subsubsection{Tower Reconfiguration with Tuned Settings}
\label{sec:V_production}
Finally, we run the tower reconfiguration with the tuned settings ($n{=}4$, $k{=}8$ with the $150$ information edges, $0.21$~s hold), shown in orange in Fig.~\ref{fig:codesign_progress}. The magnetic interaction forces realize the commanded forces to a median relative error of $4\times10^{-9}$, and the resulting trajectory follows the averaged model (\ref{NODA_unified_averaged_interaction}) within $0.40$~mm for the median spacecraft over the run 
while the worst tracking error of the two loops differs by less than $1\%$.
 The wheel exchange achieves uniform distribution of \eqref{NODA_uniform_RW_distribution}, which validates Theorem~\ref{NODA_theorem_overlapping_groups} (Fig.~\ref{fig:codesign_progress}(e)).
We compare this result with the inverse allocation for $K\leq$~10, shown in pink in Fig.~\ref{fig:codesign_progress} under the same settings. The inverse allocation carries $2.3$ times the disturbances $\bm d_{j\leftarrow k}^{2\omega_p}$ and $\bm d_{j\leftarrow k}^{pq}$ of
\eqref{Nonlinear_NODA_MMH_multifrequency_complete_interaction} and
$2.4$ times the power. Over the sweep $n\in[2,7]$, the duality-gap ratio $\gamma_W$ of Proposition~\ref{prop:matching_certificate} is $1.0$ at $n{=}2$ \cite{takahashi2026power} and $1.11$--$1.43$ for $n\in[3,7]$ with $\max<2.2$ (Table~\ref{tab:duality_gap}), in agreement with \cite{takahashi2024neural}. 
\begin{table}[tb!]
\centering
\caption{Empirical certificate
$\gamma_W=W_{\mathrm{ODA}}/J_{\mathrm{dual}}$.}
\label{tab:duality_gap}
\begin{tabular}{|l|c|c|c|c|c|c|}
\hline
& $n{=}2$ & $n{=}3$ & $n{=}4$ & $n{=}5$ & $n{=}6$ & $n{=}7$ \\ \hline
Count & $130598$ & $97483$ & $60378$ & $37857$ & $22966$ & $13291$ \\ \hline
Mean & $1.000$ & $1.110$ & $1.218$ & $1.315$ & $1.382$ & $1.431$ \\ \hline
Std & $0.000$ & $0.128$ & $0.190$ & $0.228$ & $0.237$ & $0.239$ \\ \hline
Max & $1.000$ & $1.533$ & $1.812$ & $2.006$ & $2.087$ & $2.133$ \\ \hline
\end{tabular}
\end{table}
\begin{remark}
\label{rem:radius}
The $d^{-4}$ decay decouples distant groups, and the cross-group force and torque disturbances of \eqref{eq:MIG_AC_cut_penalty} on the densest frame evaluated within $R^\star=5$~m reach $0.998$ and $0.994$ of their values at $8$~m, whereas $R^\star=3$~m leaves $0.994$ and $0.958$. We drop interactions beyond $R^\star=5$~m.
\label{rem:carriers}
Carriers are distinct within $R^\star$ and reused beyond it.
\end{remark}
\subsection{Orbital Reconfiguration with Frozen Approximation}
\label{sec:V_orbit}
\label{sec:V_orbit_case}
We next consider spacecraft swarm reconfiguration on an Earth orbit for the $64$ spacecraft of $1.33$~kg, each with three orthogonal coils of $200$ turns, $20~\Omega$, and $5$~cm radius, so that a dipole amplitude $\mu$ costs $R\mu^2/(2N^2A^2)=4.1$~W per (A\,m$^2$)$^2$. These move from a staggered $4^3$ cube onto an $8^2$ grid
of $0.3$~m pitch on a circular orbit at $500$~km altitude and
inclination $30^\circ$ with the linearized dynamics
\cite{takahashi2025distance}. Each spacecraft starts with $5^\circ$ in
attitude error and at rest in angular rate. The amplitudes are fixed during $T_u=59.8$~s, with the carriers on multiples of $4.75$~Hz. 
\subsubsection{Accuracy of the Frozen Approximation}
Our orbital reconfiguration takes $24$ orbits, including the initial formation for one orbit, transfers over $20$ orbits, and the final grid for three orbits. Since one orbit $T_{\mathrm{orb}}$ takes $94.7$~min, neighbors move
relative to each other by $q_e\approx2\pi T_u/T_{\mathrm{orb}}\approx7\%$
of their distance $d_e$ per $59.8$~s hold, with $19\%$ at most over
the magnetic edges. This keeps the position bound of
\eqref{eq:frozen_bounds} below $81$~mm on this trajectory and $18$~mm at the median hold, with $k_{\max}=15$ at $k=6$, the largest edge force of the minimum-norm routing $\bar F=56$~nN at the median hold, and $\lambda_2=1.15$. The measured difference between the frozen and carrier-resolving loops is $0.59$~mm at the median hold and $2.0$~mm at most over all $2261$ holds, with the attitudes agreeing to $2\times10^{-5}$ in the attitude parameter.
\subsubsection{Coordination Settings for MIG, CNS, and SKC}
For MIG and CNS, we set $n{=}4$ and $k{=}6$ with $26$ routed information edges, which raise $\lambda_2$ from $0.23$ to $1.03$ at the initial formation and from $0.28$ to $1.15$ at the median hold. For SKC, an offset along the nadir direction meets a negative stiffness $f_{\mathrm{orbit}}$ per unit mass and drifts along the orbital motion \cite{takahashi2025distance}, and the gains of \eqref{centralized_controller} must therefore exceed $f_{\mathrm{orbit}}$ on the slowest mode and respect the hold delay on the stiffest, which gives the window $\frac{f_{\mathrm{orbit}}\,m_{\mathrm s}}{\lambda_2}<k_p<\frac{m_{\mathrm s}}{\lambda_NT_u^{2}}$. The $26$ information edges widen this window from $1.5\times$ to $5.7\times$. To satisfy these constraints, we select $k_p=1.6\times10^{-5}$~N/m and $k_v=1.0\times10^{-3}$~N\,s/m inside $4.8\times10^{-6}<k_p<2.8\times10^{-5}$~N/m, which gives $(k_v/m_{\mathrm s})\lambda_NT_u=0.6$. Likewise, the hold delay bounds the attitude gains by $k_\sigma<J_{\mathrm s}/T_u^2$ and $k_\omega<J_{\mathrm s}/T_u$ for the spacecraft inertia $J_{\mathrm s}$, and we set $k_\sigma=0.6J_{\mathrm s}/T_u^2=3.7\times10^{-7}$~N\,m and $k_\omega=0.6J_{\mathrm s}/T_u=2.2\times10^{-5}$~N\,m\,s for $J_{\mathrm s}=2.2\times10^{-3}$~kg\,m$^2$.
\subsubsection{Numerical Results}
The calculation result in Fig.~\ref{fig:orbit_snapshots} holds the formation to
$0.61$~mm in the median 
under a commanded force of $64$~nN in the median and $1.1$~$\mu$N at most.
The attitude of the worst spacecraft falls below $1^\circ$ within $15$~min of the $20$~min slew reference, and after the slew the swarm
median settles at $6\times10^{-5}$~degrees with the worst spacecraft
below $0.005^\circ$ on every hold (Fig.~\ref{fig:orbit_snapshots}(e)). The
wheels absorb the angular momentum that the reconfiguration releases
and converge to the uniform share of \eqref{NODA_uniform_RW_distribution}, $\|\bm L-\bm L_d\|/N_{\mathrm w}=0.628$~mNms. At the end of the mission the median wheel holds this share and the farthest one $0.636$~mNms,
$1.3\%$ apart (Fig.~\ref{fig:orbit_snapshots}(f)). The coil power of
Fig.~\ref{fig:orbit_snapshots}(d) is $74$~mW per spacecraft averaged over the mission and $3.4$~W at the peak of the worst spacecraft.
\begin{figure}[tb!]
\centering
\subfloat[$t=0$.]{\includegraphics[width=0.32\linewidth]{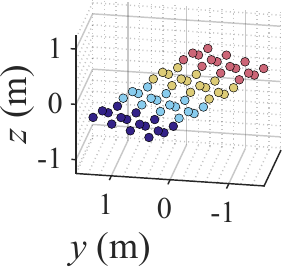}}\hfill
\subfloat[$t=0.78$~d.]{\includegraphics[width=0.32\linewidth]{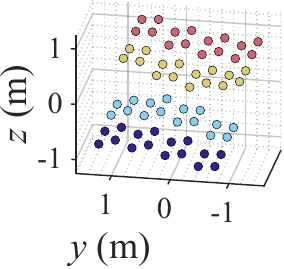}}\hfill
\subfloat[$t=1.56$~d.]{\includegraphics[width=0.32\linewidth]{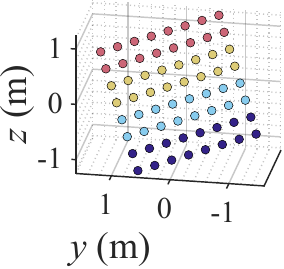}}\\
\vspace{0.02cm}
\subfloat[Coil power.]{\includegraphics[width=0.32\linewidth]{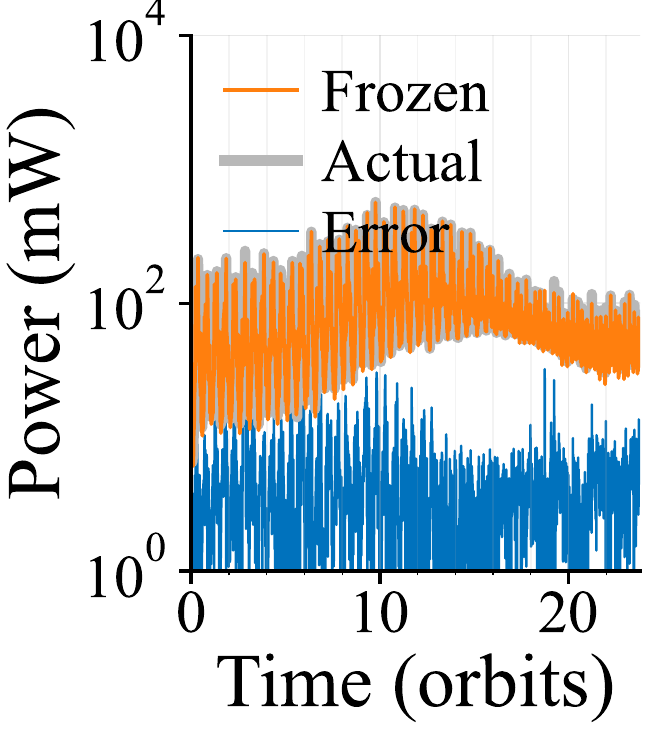}}\hfill
\subfloat[Attitude error.]{\includegraphics[width=0.32\linewidth]{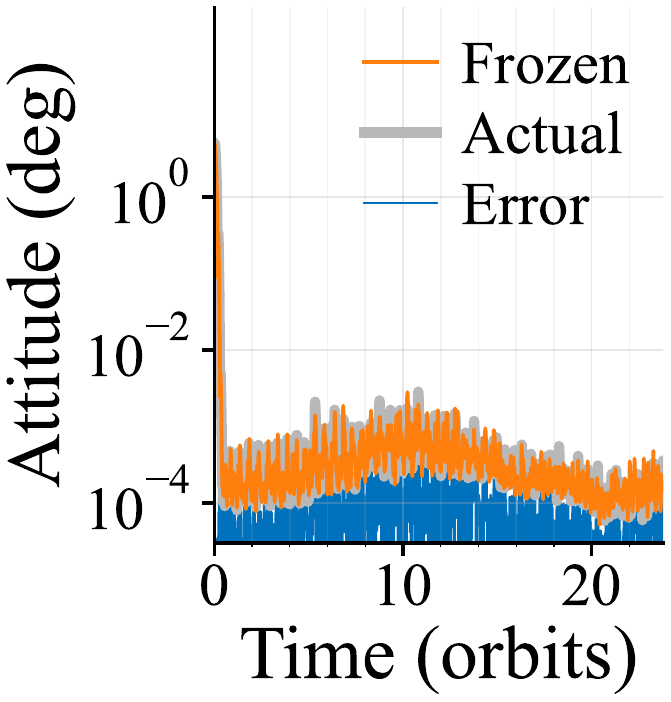}}\hfill
\subfloat[RW momentum.]{\includegraphics[width=0.32\linewidth]{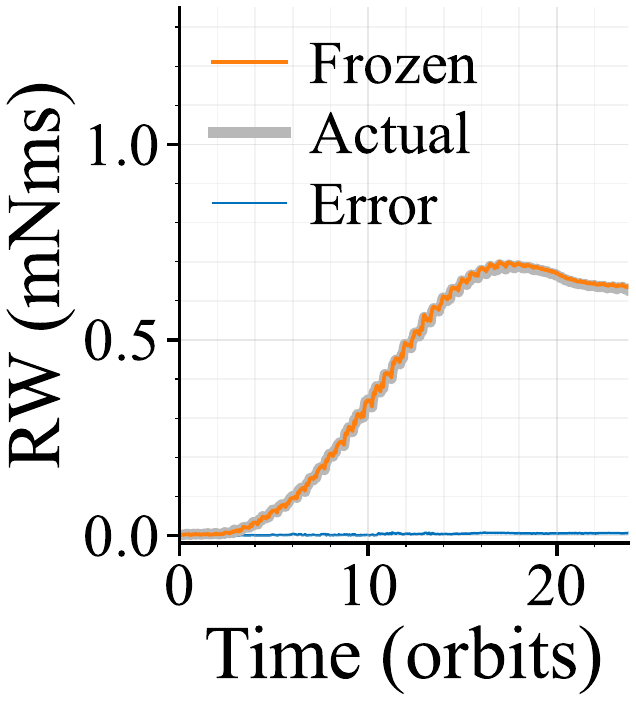}}
\caption{Orbital reconfiguration of $64$ spacecraft ($y$ along-track, $z$ normal, viewed from the radial direction).}
\label{fig:orbit_snapshots}
\label{fig:orbit_2weeks}
\end{figure}
\section{Conclusions}
\label{sec:conclusions}
This paper developed a decentralized power-optimal coordination
framework for magnetically actuated spacecraft swarms. The designed commands of overlapping groups converge while preserving momentum. Grouping reduces the number of carriers and
admits higher consensus gains. The power-minimizing allocation keeps
the time-averaging error small enough to hold the formation. The
approximate integration with a bounded
error carries the framework to a long-horizon orbital
reconfiguration.
\section*{Acknowledgment}
The first author would like to thank Prof. Hiraku Sakamoto at the
Institute of Science Tokyo for technical discussions. The authors
acknowledge the Tokyo Digital Twin Project for providing the open 3D building point-cloud data used for visualization in this paper. The authors used Claude to improve the English text.